    \documentclass[hidelinks,onefignum,onetabnum]{siamart250211} %\documentclass[final,floating,letterpaper]{siamltex1213}

    \usepackage{url}
    \usepackage{verbatim}
    \usepackage[titletoc]{appendix}
    \usepackage{hyperref}
    \usepackage{bm}
    \usepackage{amsmath}
	\usepackage{amsfonts}
    \usepackage{nicefrac}
    \usepackage{longtable}
    \usepackage{color}
    \usepackage{graphicx, amssymb,graphics}
    \usepackage{algorithmic,algorithm}
    \usepackage{indentfirst}

    \usepackage{paralist}

    \usepackage{graphicx}
    \usepackage{subcaption}

\newtheorem{rem}{Remark}[section]

\providecommand{\dd}{\mathop{}\!\mathrm{d}}

\newcommand{\reff}[1]{{\rm (\ref{#1})}}
	\newcommand\be {\begin{equation}}
	\newcommand\ee {\end{equation}}

\newcommand{\ve}{\varepsilon}          % varepsilon

\numberwithin{equation}{section}

\graphicspath{{..}}

\title{Intrinsic local Gauss's law preserving PIC method: A self-consistent field-particle update scheme for plasma simulations}
\date{\today}

\begin{document}
	
\author{
Zhonghua Qiao\thanks{Department of Applied Mathematics, the Hong Kong Polytechnic University, Hong Kong, China. (zhonghua.qiao@polyu.edu.hk)}~
\and
Zhenli Xu \thanks{School of Mathematical Sciences, MOE-LSC, and CMA-Shanghai, Shanghai Jiao Tong University, Shanghai 200240, China. (xuzl@sjtu.edu.cn)}~
\and
Qian Yin \thanks{Corresponding author. Department of Applied Mathematics, the Hong Kong Polytechnic University, Hong Kong, China. (qqian.yin@polyu.edu.hk)}~
\and
Shenggao Zhou\thanks{School of Mathematical Sciences, MOE-LSC, and CMA-Shanghai, Shanghai Jiao Tong University, Shanghai 200240, China. (sgzhou@sjtu.edu.cn)}
}

\maketitle

\begin{abstract}	
In order to perform physically faithful particle-in-cell (PIC) simulations, the Gauss's law stands as a critical requirement, since its violation often leads to catastrophic errors in long-term plasma simulations.
This work proposes a novel method that intrinsically enforces the Gauss's law for the Vlasov-Ampère/Vlasov-Poisson system without requiring auxiliary field corrections or specialized current deposition techniques.
The electric field is managed to get updated locally and consistently with the motion of particles via splitting the motion into sub-steps along each dimension of the computational mesh. To further obtain a curl-free electric field, a local update scheme is developed to relax the electric-field free energy subject to the Gauss's law. The proposed method avoids solving the Poisson's or Amp\`ere's equation, resulting in a local algorithm of linear complexity for each time step which can be flexibly combined with various temporal discretization for particle motion in PIC simulations. Theoretical analysis verifies that the proposed method indeed maintains the discrete Gauss's law exactly.
Numerical tests on classical benchmarks, including the Landau damping, two-stream instability and Diocotron instability, demonstrate the key advantages of the proposed method. %(i) elimination of accumulated Gauss law violations, and (ii) seamless compatibility with momentum/energy conservation strategies, and (iii) expection to reduction in computational cost compared to hybrid Poisson-PIC approaches when increasing grid resolution as mesh refines.
It is expected that the local nature of the proposed method makes it a promising tool in parallel simulations of large-scale plasmas.
% The Gauss's law stands as a critical requirement for ensuring physical faithful in particle-in-cell (PIC) simulations, in which violations often lead to catastrophic errors in long-term plasma simulations. This work presents a novel PIC framework that intrinsically enforces Gauss's law for Vlasov-Ampère (VA) systems without requiring auxiliary field corrections or specialized current deposition techniques. In the PIC method, the field calculation are dynamically coupled with  the trajectory of particle motion and a local curl-free diffusive relaxation to realize the curl free of electric fields. This procedure avoids to solve the Poisson or Amp\`ere equations, resulting in rapid calculations of the electric field.
% Meanwhile, the new local method can be coupled with other techniques to achieve additional desirable properties, like energy conservation or other specific objectives.
% Numerical tests on classical benchmarks—including Landau damping, two-stream instability and Diocotron instability—demonstrate the key advantages of the new method. %(i) elimination of accumulated Gauss law violations, and (ii) seamless compatibility with momentum/energy conservation strategies, and (iii) expection to reduction in computational cost compared to hybrid Poisson-PIC approaches when increasing grid resolution as mesh refines.
% These findings position the methodology as a robust foundation for large-scale plasma simulations.

\bigskip



\noindent
\end{abstract}

\section{Introduction}
%Problem introduction:
%Literature review:
%Advantages and Contributions and Further improvements:
The Vlasov equation is a kinetic equation describing the evolution of charged particles under electromagnetic fields~\cite{Colonna2022plasma}. It differs slightly from the Boltzmann equation~\cite{Hagelaar2005Plasma}, neglecting short-range collisions and emphasizing on interactions governed by self-consistent fields. In the complete electromagnetic formulation, the Vlasov equation is coupled with the full set of Maxwell's equations, where the electric and magnetic field can be solved using charge densities and currents generated by particle motions. Under electrostatic approximations, this framework reduces to two equivalent formulations~\cite{chen2011JCP:VA,Li2023JCP}, the Vlasov-Poisson system, computing electric field via Poisson's equation, and Vlasov-Amp\`ere system, deriving electric field from Amp\`ere's equation under local charge continuity equation, both neglecting magnetic effects.

While analytical solutions exist only for idealized systems, such as the linear waves; therefore, numerical methods become necessary in order to study the phenomenon governed by the Vlasov equation. Overall, these numerical methods can be broadly classified into two categories: Eulerian approaches, which involve direct solutions of the Vlasov equation, and particle-based methods, in which the distribution is approximated through discrete macro-particles.  The Eulerian methods, including the discontinuous Galerkin~\cite{Cheng2014JCP}, finite volume~\cite{Filbet2001SINUM}, semi-Lagrangian~\cite{Liu2023JCPsemiLagrange,Qiu_Shu_2011CiCP}, moment methods~\cite{Taitano2024arxiv} and so on~\cite{Guo2024SISC,Kirchhart2024SISC}, are suitable for high-precision phase-space analysis, whereas the Particle-in-Cell (PIC) method~\cite{Birdsall2018,Kormann2024SISC,Lapenta2015,Muga1970PIC}, belonging to particle-based methods, dominates high-dimensional, large-scale simulations due to its computational scalability.
Some works use phase-space remapping to reduce the numerical noise in PIC simulations to obtain higher accuracy~\cite{Myers2017SISC,Wang2011SISC}.
There are also many attempts to prove the convergence of particle methods for the Vlasov equation coupled with the Poisson's equation~\cite{Cottet1984SINUM,Pinto2016SINUM}. %\zhou{Should describe VM VA VP briefly somewhere.}

Preserving critical mathematical and physical properties in PIC simulations is essential and highly desirable for ensuring the accuracy and reliability of numerical results, especially in long-term simulations. For instance, asymptotic-preserving (AP) schemes~\cite{Degond_JCP2010,Ji2023JMP} that are able to maintain cross-scale consistency have been shown to be indispensable in resolving multiscale coupling challenges when the Debye length approaches the quasi-neutral limit. Also, the total energy, which comprises contributions from electric field energy and kinetic energy, should theoretically remain conservative in continuous systems. But finite-grid instabilities can induce unphysical heating or cooling phenomena. This has motivated extensive development of energy-conserving schemes~\cite{chen2015CPC,Chacon2016JCP,Chen2020JCP,Lapenta2017JCP,Ricketson2023CPC,Ricketson2024explicit}. Additionally, momentum-conserving methods~\cite{Cohen1982JCP,Mason198JCP} are desirable as well. Among the various conservation requirements, mass conservation and the Gauss's law form the cornerstones of robust PIC simulations, and have been critically addressed in the aforementioned studies. Significantly, as shall be shown in Section 2.3, the design of local mass conservation inherently helps the preservation of the Gauss's law. As highlighted by Anderson \emph{et al.}~\cite{Anderson2020JCP}, failure to enforce the Gauss's law can lead to catastrophic consequences in PIC simulations.

This study proposes an intrinsic local Gauss's law-preserving PIC framework that can be strategically combined with additional conservation-enforcing methods to achieve comprehensive conservation properties for the Vlasov-Amp\`ere systems.
%As the electrostatic limit of the Vlasov-Maxwell's system, the VA systems under certain assumptions can be shown to be equivalent to the Vlasov-Poisson (VP) formulation~\cite{chen2011JCP:VA,Li2023JCP}.
Recent progress on the enforcement of the Gauss's law, which is systematically reviewed in Section 2.3, fall into two distinct categories: field correction methods~\cite{Boris1970,Langdon1992CPC,Marder1987JCP,Munz2000JCP} that post-process the electric field to eliminate inconsistencies in the Gauss's law, and charge-conserving methods~\cite{Morse1971PF,Sokolov2013CPC,Villasenor1992CPC} that enforce discrete continuity equations through carefully designed current densities. This work proposes a novel Gauss's law-preserving framework that fundamentally differs from these conventional approaches.
% by local updates of the electric field.
% using local updates of electric field to maintain the discrete Gauss's law after charge motion.
The proposed algorithm achieves this through a unified field-particle advancement scheme that simultaneously updates electric fields and particle positions while strictly maintaining the discrete Gauss's law
% $\nabla \cdot \varepsilon \bm{E}=\rho$
up to machine precision.
More specifically,
the electric field is locally updated to account for electric fluxes induced by the particle motion along the trajectory of motion, and then these updates further diffuse to the whole area through a local curl-free relaxation approach aiming to minimize the electric field free energy.
The whole process eliminates explicit dependence on the Poisson or Amp\`ere solvers through its self-consistent treatment of field propagation and particle dynamics. Such type of local algorithms
is firstly proposed for  Monte Carlo simulations~\cite{M:JCP:2002,MR:PRL:2002,RM:JCP:2004MC}, as well as the molecular dynamics simulations~\cite{FH:PRE:2014,FXH:JCP:14,PD:JPCM:2004,RM:PRL:2004} for electrostatic systems. Later, it was further generalized to solve continuum electrostatic models, such as the Poisson--Boltzmann equation~\cite{BSD:PRE:2009,Qiao2024SISC,ZWL:PRE:2011} and Poisson--Nernst--Planck equations~\cite{Qiao2023SIAP,Qiao2022NumANP}.
The distinctive features on the preservation of the Gauss's law and inherent locality makes the proposed method particularly advantageous for plasma simulations where simultaneous conservation of multiple properties (mass, momentum, energy) is crucial.

The rest of this paper is organized as follows. Section 2 provides an overview of Vlasov type of models and PIC methods, and reviews related works on the treatment of electric fields that preserve the Gauss's law. In Section 3, we elaborate on the proposed local Gauss's law preserving PIC methods. Some numerical experiments are conducted in Section 4. Finally, conclusions are summarized in Section 5.

% \section{Vlasov-Amp\`ere equations and particle-in-cell method}
\section{Model and numerical methods}
\subsection{Vlasov-Amp\`ere equations}
The Vlasov equation, often used to describe inertial confinement fusion, astrophysics and so on, is given by
	$$
		\frac{\partial f_s(\boldsymbol{x}, \boldsymbol{v}, t)}{\partial t} +\boldsymbol{v} \cdot \nabla_x f_s(\boldsymbol{x}, \boldsymbol{v}, t)+\bm{F}(\bm{x},t) \cdot \nabla_v f_s(\boldsymbol{x}, \boldsymbol{v}, t)=0,
	$$
	where $f_s(\bm{x}, \bm{v}, t)$ represents the distribution of particles of species $s$ at a specific position $\bm{x}$, velocity $\bm{v}$ and time $t$. Here $\bm{F}(\bm{x},t)$ is the external force. Generally, the force field includes electrostatic and Lorentz force, i.e. $\bm{F}(\bm{x},t)=q_s(\bm{E}+\bm{v}\wedge \bm{B})/m_s $, where $q_s$ is charge carried by particles, $m_s$ is mass, and the electric field $\bm{E}$ and magnetic field $\bm{B}$ are described by the Maxwell's equation
    \begin{align}
    		&\frac{\partial \bm{E}}{\partial t}-c^2 \nabla \times \bm{B}=-\frac{\bm{J}}{\epsilon_0}, \label{Ampere}\\
    		&\frac{\partial \bm{B}}{\partial t}+\nabla \times \bm{E}=0, \label{Faladay}\\
    		&\nabla \cdot \bm{E}=\frac{\rho}{\epsilon_0}, \label{Gauss}\\
      & \nabla \cdot \bm{B}=0. \label{Thomas}
    	\end{align}
     Here $c$ is the speed of light, $\epsilon_0$ is the vacuum permittivity, $\rho(\bm{x},t)$ the charge density  and $\bm{J}(\bm{x},t)$ the current density are expressed by $$\rho(\boldsymbol{x}, t)=\sum_s q_s \int_{\Omega_v} f_s \dd \boldsymbol{v},~~~\boldsymbol{J}(\boldsymbol{x}, t)=\sum_s q_s \int_{\Omega_v} f_s \boldsymbol{v} \dd \boldsymbol{v},$$
     respectively.

    In this work, we consider the electrostatic limit.
    % If the electrostatic approximation of the Maxwell's equation is considered,
    The above Vlasov-Maxwell model can be approximated by the Vlasov-Poisson (VP) model, in which the force becomes $\bm{F}(t,\bm{x})=q_s \bm{E}/m_s$ with $\bm{E}=-\nabla \phi$. The electric potential $\phi$ is governed by the Poisson's equation
    $$-\nabla \cdot \nabla \phi(\boldsymbol{x}, t)={\rho}/{\epsilon_0}.$$

    By integrating the Vlasov equation with respect to the velocity field and summing over $s$, one gets the charge continuity equation
	\begin{equation} \label{charge}
	\frac{\partial \rho(\boldsymbol{x}, t)}{\partial t}+\nabla \cdot \boldsymbol{J}(\boldsymbol{x}, t)=0.
	\end{equation}
	By the Gauss's law~\reff{Gauss}, one gets
	$$
    \nabla \cdot \left[ \epsilon_0	\frac{\partial \bm{E}}{\partial t}+\boldsymbol{J}\right]=0.
	$$
    Therefore, the electric field $\bm{E}$ can be described by the Amp\`ere's equation~\reff{Ampere}
    \begin{equation}\label{Ampere2}
        \epsilon_0	\frac{\partial \bm{E}}{\partial t}+\boldsymbol{J}=\bm{Q},
    \end{equation}
where $\bm{Q}$ is an extra freedom satisfying a divergence free condition $\nabla \cdot \bm{Q}=0$. Coupling \reff{Ampere2} and the curl-free condition $\nabla \times \bm{E}=0$ with the Vlasov equation gives the Vlasov-Amp\`ere (VA) model.
	% Integrate both sides of above equation in space, we can get Amp\`ere equation~\reff{Ampere2}.

	In the continuum, the VP and VA formulations can be shown to be equivalent~\cite{chen2011JCP:VA,Li2023JCP}. Above derivation also reveals that the charge continuity equation~\reff{charge} and Gauss's law~\reff{Gauss} play critical roles in conversion between the VP and VA formulations.
    % In the following section, we will further notice the importance of these two laws.

We perform non-dimensionalization, as detailed in~\cite{Li2023JCP}, and still take
% and we omit,
the original variables without causing any ambiguity. For the sake of descriptive convenience, we consider the Vlasov equation with only one specie, i.e. electron:
 \begin{equation}
		\frac{\partial f(\boldsymbol{x}, \boldsymbol{v}, t)}{\partial t} +\boldsymbol{v} \cdot \nabla_x f(\boldsymbol{x}, \boldsymbol{v}, t)-\bm{E}(\bm{x},t) \cdot \nabla_v f(\boldsymbol{x}, \boldsymbol{v}, t)=0. \label{Vlasov}
\end{equation}
The non-dimensionalized Poisson's equation and Amp\`ere's equation are given by
$$
-\nabla \cdot \lambda^2 \nabla \phi(\boldsymbol{x}, t)=\rho={n_0-n},~~ \bm{E}=-\nabla \phi,
$$
and
$$
\lambda^2	\frac{\partial \bm{E}}{\partial t}+\boldsymbol{J}=\bm{Q}, ~~\nabla\cdot \bm{Q}=0,~~\nabla\times \bm{E}=0,
$$
respectively,
where $\lambda$ is ratio of the Debye length to the characteristic length of the system under consideration, $n_0$ is the density of background ion, and
$$n(\boldsymbol{x}, t)=- \int_{\Omega_v} f_s \dd \boldsymbol{v},~\hbox{and}~~\boldsymbol{J}(\boldsymbol{x}, t)=- \int_{\Omega_v} f_s \boldsymbol{v} \dd \boldsymbol{v}.$$

\subsection{Particle-in-cell method}
The PIC method is a widely used numerical technique for simulating the Vlasov equations for plasma dynamics. It combines the motion of charged particles with the evolution of electromagnetic fields, making it a powerful tool for tackling the high-dimension challenge in the simulation of a variety of plasma phenomena. In the PIC method, the distribution function $f$ is approximated by
\begin{equation}\label{fapp}
f(\boldsymbol{x}, \boldsymbol{v}, t)=\sum_{p=1}^{N_p} w_p S\left(\boldsymbol{x}-\boldsymbol{x}_p\right) \delta\left(\boldsymbol{v}-\boldsymbol{v}_p\right),
\end{equation}
where $N_p$ is the number of computational particles, $\delta\left(\boldsymbol{v}-\boldsymbol{v}_p\right)$ is the Dirac delta function, and $w_p$ is the weight of the $p$-th particle defined by the ratio of numbers of real physical particles to computational particles and is often treated as a constant in collisionless PIC simulations. In addition, $S\left(\boldsymbol{x}-\boldsymbol{x}_p\right)$ is the shape function that has a particular compact support, is symmetric, and integrates to 1, i.e. $\int_{\Omega_x}S(\bm{x})d\bm{x} = 1$. It can be chosen as the B-spline, Gaussian functions or others. The most commonly used shape function in 1D is the ``tent" function, defined as $S(x)=\max\{0,1-|x|/h\}/h$, where $h$ is the grid spacing. The shape function in higher dimensions is obtained by taking the product of 1D ones, e.g., $S(x,y,z)=S_x(x)S_y(y)S_z(z)$.
Substituting the particle approximation~\reff{fapp} into the Vlasov equation~\reff{Vlasov}, one can obtain the evolution equation for the particles from the first-order momentum of the Vlasov equation~\cite{Lapenta2015}:
\begin{align}
		&\frac{\dd \bm{x}_p}{\dd t} =\bm{v}_p , \label{newton1}\\
		& \frac{\dd \bm{v}_p}{\dd t} = - \bm{E}(\bm{x}_p), \label{newton2}
	%	& \epsilon_0 \nabla \cdot \epsilon_{s} (\bm{x}) \bm{E}(\bm{x}) =-e[\rho-N_D], \label{gauss}\\
	%	& \nabla \times \bm{E}=\bm{0}, \label{curlfree}
\end{align}
for $p=1,\cdots,N_p$.
The mobile particles follow the above Newton’s law, which imparts a “granular”
nature. Meanwhile, the potential and electric field are represented by values on a regular
mesh of grid points, necessitating the projection of discrete values defined on the grid points
onto the distribution of particles, and vice versa.
Densities defined on the mesh are calculated by assigning the particle charge to nearby mesh
points:
\begin{equation}\label{density}
n_h(t)=\sum_{p=1}^{N_p}\omega_pS\left(\bm{x}_h-\bm{x}_p (t)\right),
\end{equation}
where the subscript $h$ indicates that the quantity is defined on the mesh.
Similarly, for currents,
\begin{equation}\label{current}
\bm{J}_h(t)=-\sum_{p=1}^{N_p}\omega_p S\left(\bm{x}_h-\bm{x}_p (t)\right)\bm{v}_p(t).
\end{equation}
With densities $n_h(t)$ and currents $\bm{J}_h(t)$, field $\bm{E}_h$ defined on mesh can be computed using the Poisson's equation or Amp\`ere's equation. Then,
potentials and fields at particle positions are further obtained by interpolating the mesh-defined values by
\[
\boldsymbol{E}(\bm{x}_p)=|\Delta V|\sum_{h}  \boldsymbol{E}_h(t) S \left(\boldsymbol{x}_h-\boldsymbol{x}_p(t)\right),
\]
where $|\Delta V|$ is the volume of a single cell.

In PIC simulations of plasma, it is necessary to obtain
accurate spatial distribution of potential
and electric fields, ensuring proper treatment of particle dynamics. The motion of particles
can in turn cause global changes of the fields in the system. Therefore, even minute
particle movement requires recalculating the Poisson’s equation globally at each time step,
which accounts for the vast majority of the total computational effort. The update method proposed in this work ensures rapid and accurate acquisition of the desired electric fields by concurrently updating the particle positions and the electric field, while rigorously adhering to the Gauss's law.

\subsection{Gauss's law and charge conservation}~\label{s:Gauss}
The Gauss’s law~\reff{Gauss}, which states that the divergence of the electric field is proportional to the charge density, is a fundamental principle of electromagnetism. In PIC simulations, accurate enforcement of the Gauss’s law ensures that the electric field is correctly coupled with the charge distribution, being essential for modeling the interactions between particles and fields.
In this section, we emphasize the importance of the Gauss's law in PIC simulations and review various methods from the literature devoted to satisfy the Gauss's law, highlighting the advantage of our algorithm in naturally fulfilling the law.

One class of methods is to correct the electric field to enforce the Gauss's law by adding a correction term.
The well-known Boris correction~\cite{Boris1970,Birdsall2018} introduces an extra term $\delta \phi$
so that the corrected electric field  $\bm{E}_{\text{correct}}=\bm{E}-\nabla \delta \phi$  satisfies  the Gauss's law $\nabla \cdot \varepsilon
\bm{E}_{\text{correct}}=\rho$. Hence, solving the Poisson's equation,
$$\nabla \cdot \ve \nabla \delta \phi =\nabla \cdot \ve \bm{E}-\rho,$$
is needed in the Boris correction method.
Later, Marder~\cite{Marder1987JCP} and Langdon~\cite{Langdon1992CPC} proposed an alternative type of method through replacing the Poisson solver by a local update of electric fields
$$
\bm{E}_{\text{correct}}^{n+1}=\bm{E}^{n+1}+\Delta t \nabla[d (\nabla \cdot \ve \bm{E}-\rho)],
$$
where $d$ is the diffusion parameter satisfying a stability condition~\cite{Langdon1992CPC,Mardahl1997CPC}. Later, Munz \emph{et al.} generalized this idea into a framework of the Lagrange multiplier method~\cite{Munz2000JCP}. Recently, Chen and T\'{o}th~\cite{Chen2019JCP} proposed a series of particle position correction methods to realize the Gauss's law, taking into account the trade-off between computational cost and the tolerance error of the Gauss's law. Different from the previous, such methods correct the displacement of particles to ensure that the energy conservation property is not violated, especially when combined with energy conserving semi-implicit methods~\cite{Lapenta2017JCP}.

%Violations of charge conservation can lead to nonphysical results, such as artificial charge accumulation or loss, which can significantly affect the simulation outcomes.	
The global charge conservation law, which describes that the total charge in the system remains constant, is crucially important as well. The particle method used in this work naturally preserves this global law. In literature, density or current is designed to
% In many researches, they design density or current to
fulfill the discrete version of the local charge conservation equation~\reff{charge} and therefore satisfy the Gauss's law. The rationale behind this is detailed as follows.
% the theoretical guarantee of which can be derived as follows.
Consider certain temporal discretization
 $$
	\epsilon_0 \frac{\bm{E}^{n+1}_h-\bm{E}^{n}_h}{\Delta t}+\bm{J}^*_h=\bm{Q}^*_h,
 $$
 where the superscript $*$ indicates an either implicit or explicit discretization time step. Taking spatial divergence on the both sides yields
 $$\epsilon_0  \nabla_h \cdot \frac{\bm{E}^{n+1}_h-\bm{E}^{n}_h}{\Delta t}+\nabla_h \cdot \bm{J}^*_h=0,
	$$
where $\nabla_h$ is the discrete analog of the operator $\nabla$ and shall be detailed in Section~\ref{ss:dis}.
By the discrete charge conserving law
 \begin{equation}\label{discharge}
     \frac{\rho^{n+1}_{h}-\rho^{n}_{h}}{\Delta t} +\nabla_h \cdot \bm{J}^*_h=0,
 \end{equation}
one can obtain
 % then we can get
	$$
	\epsilon_0  \nabla_h \cdot \frac{\bm{E}^{n+1}_h-\bm{E}^{n}_h}{\Delta t}+\frac{\rho^{n+1}_h-\rho^{n}_h}{\Delta t}=0.
	$$
Therefore, one can readily derive that provided by the Gauss's law at time $n\Delta t$ is valid, the discrete charge conserving law and the Gauss's law at time $(n+1)\Delta t$ are equivalent.
	% meaning that Gauss's law is satisfied at time $(n+1)\Delta t$ provided by Gauss's law at time $n\Delta t$ and
 % discrete charge conserving law
 % \begin{equation}\label{discharge}
 %     \frac{\rho^{n+1}_{h}-\rho^{n}_{h}}{\Delta t} +\nabla_h \cdot \bm{J}^*_h=0.
 % \end{equation}
	Substituting the expressions~\reff{density} and \reff{current} with the shape function into~\reff{discharge}, one gets
	$$
	\frac{S(\bm{x}_{h}-\bm{x}_p^{n+1})-S(\bm{x}_{h}-\bm{x}_p^{n})}{\Delta t} +\nabla_h \cdot [\bm{v}^*_p S(\bm{x}_{h}-\bm{x}_p^{*})]=0,
	$$
    which is the basis for most ``current weighting" methods
~\cite{Morse1971PF,Villasenor1992CPC,Sokolov2013CPC}.
% are based on this equation.
% Additionally, the charge conserving is violated when a particle crosses the cell boundary, so~\cite{chen2011JCP:VA,Chacon2016JCP} split the particle motion by sub timestep to restrict the particle moving within one cell, thus fulfilling exact charge conservation~\reff{discharge}.
In PIC simulations, the charge conserving law may get violated when a particle crosses cell boundaries. The particle motion can be split into sub-steps to restrict the particle moving within one cell, so that the exact charge conservation can be achieved~\cite{chen2011JCP:VA,Chacon2016JCP}.
%The alternating-order interpolation method proposed by Sokolov~\cite{Sokolov2013CPC} combines different orders of interpolation for different physical variables, which significantly improves the accuracy of charge conservation.

\section{Local Gauss's law preserving PIC method}
\subsection{Discretization}\label{ss:dis}
%The derived particle motion equations~\reff{newton1}--~\reff{newton2} is the simple Newton's form.
The classical temporal discretization for particle motion  equations~\reff{newton1}--~\reff{newton2} in the PIC is the so-called ``leapfrog'' scheme, where positions are defined at integer values of the time-step, i.e., $\bm{x}_p^{m} = \bm{x}_p(m\Delta t)$, while velocities are defined at
half-integer values: $\bm{v}_p^{m+1/2} = \bm{v}_p\left(\left(m + 1/2\right)\Delta t\right)$.
They are alternately updated in the following manner:
\begin{align}
	&\bm{x}^{m+1}=\bm{x}^m+{\Delta t} \bm{v}^{m+\frac{1}{2}}, \label{DisNewton1}\\
	&\bm{v}^{m+\frac{3}{2}}=\bm{v}^{m+\frac{1}{2}}-{\Delta t } \bm{E}^{m+1}(\bm{x}_p),\label{DisNewton2}
\end{align}
where $\bm{E}^{m+1}(\bm{x}_p)=\bm{E}(\bm{x}_p^{m+1})$ is the electric field interpolated using grid data. The initial velocity at a half time step is obtained using an explicit scheme
\[
\bm{v}^{\frac{1}{2}}=\bm{v}^{0}-\frac{1}{2}{\Delta t } \bm{E}^{0}(\bm{x}_p).
\]
Such staggered updating of positions and velocities over time steps gives its name leapfrog. Meanwhile, this explicit scheme has several advantages, such as second-order accuracy, computational simplicity, and preservation of phase space volume, making it suitable for long-term simulations of Hamiltonion systems.

For spatial discretization, we adopt the standard finite-difference Yee mesh~\cite{Yee1966,Yee1988}, in which scalar quantities, such as potential and charge density, are discretized on grid points of cells, while the electric fields are discretized at the center of cell edges.  Consider a square computational domain $\Omega_x=[0,L]\times[0,L]$ with periodic boundary conditions. The domain is covered with a uniform mesh with grid spacing $h=L/(N-1)$ in each dimension. Specifically, let $n_{i,j}~(i,j=0,\cdots,N-1)$ be the numerical approximation of density $n$ on the grid point $(ih, jh)$, and $E_{i+1/2,j}$ and $E_{i,j+1/2}$ are discrete analogues of the projected values of electric fields at midpoints of the edges $\left(\left(i + 1/2\right)h, jh\right)$
 and $\left(ih,\left(j + 1/2\right)h\right)$,  respectively.
%\begin{figure}[htbp]
 %        \centering
%		\includegraphics[scale=2]{Yee_cell.png}
%		\caption{Yee lattice: the placement of electric and magnetic fields on a staggered grid.}
%\end{figure}
With the Yee mesh, the Gauss's law is approximated by
\begin{equation}\label{DiscreteGauss}
\frac{E_{i+\frac{1}{2},j}-E_{i-\frac{1}{2},j}}{ h} +\frac{E_{i,j+\frac{1}{2}}-E_{i,j-\frac{1}{2}}}{h} =\frac{n_0-n_{i,j}}{\lambda^2}.
\end{equation}
It has been shown that such spatial discretization has second order accuracy for the electric field in both $L^{\infty}$ and $L^2$ norms~\cite{Li2024convergence}.

% The Yee mesh is critical for our local Gauss's law satisfying method.
% Our designed numerical method is based on this equation, updating the electric field after each change in position in a timely manner to satisfy the most important property--Gauss's law.

\subsection{Particle motion with electric-field update}

As demonstrated in Section~\ref{s:Gauss}, it is crucial to satisfy the Gauss's law for long-term simulations. The change in the electric field depends on the evolution of charge density, which in turn determines particle motion. We here propose a novel algorithm that updates the electric field simultaneously with particle motion, ensuring that Gauss's law is satisfied within machine precision. For simplicity, the algorithm is presented in 2D and can be extended in a straightforward manner to 3D.  Let particle position $\bm{x}_p=({x}_{p,x},{x}_{p,y})$ and velocity $\bm{v}_p=({v}_{p,x},{v}_{p,y})$.
% are divided to two dimensions (we consider two dimensions for simplification, although three dimensions are similar), $\bm{x}_p=({x}_{p,x},{x}_{p,y})$ and
% $\bm{v}_p=({v}_{p,x},{v}_{p,y})$, respectively.
When the particle positions get updated according to the ``leapfrog'' scheme~\reff{DisNewton1},
we update the locations of particles parallel to the links of the mesh in the following two steps
\begin{align}
&x^{m+1}_{p,x}=x_{p,x}^m+{\Delta t} v_{p,x}^{m+\frac{1}{2}}, \label{position1}\\
&x^{m+1}_{p,y}=x_{p,y}^m+{\Delta t} v_{p,y}^{m+\frac{1}{2}}.\label{position2}
\end{align}
As shown in Figure~\ref{f:move1} (a), the particle motion from $\bm{x}_p^{m}=({x}_{p,x}^{m},{x}_{p,y}^{m})$ to $\bm{x}_p^{m+1}=({x}_{p,x}^{m+1},{x}_{p,y}^{m+1})$ within one timestep is split into two orthogonal moves: first $\bm{x}_p^{m}$ to $({x}_{p,x}^{m+1},{x}_{p,y}^{m})$, and then to $\bm{x}_p^{m+1}$.

\begin{figure}[htbp]
			\centering
            \includegraphics[scale=0.56]{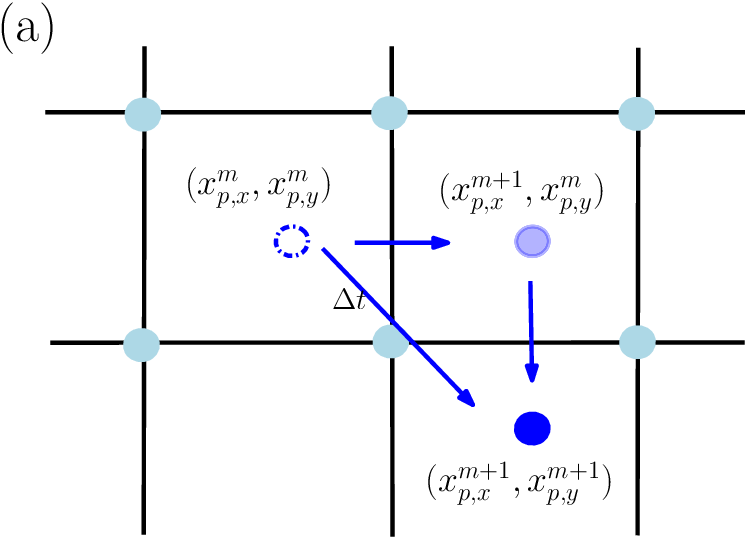}
            \vspace{2em}

			\includegraphics[scale=0.45]{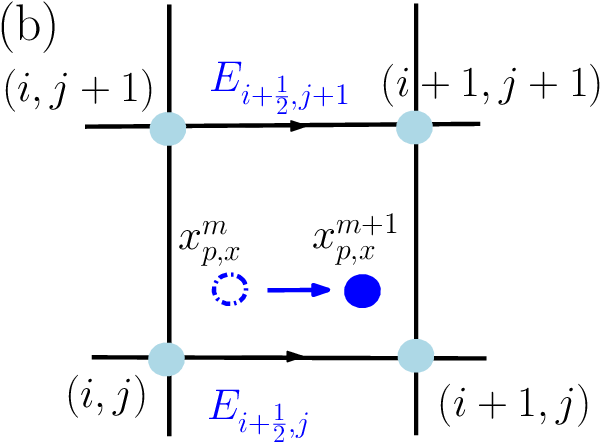}
            \hspace{1em}
            \includegraphics[scale=0.45]{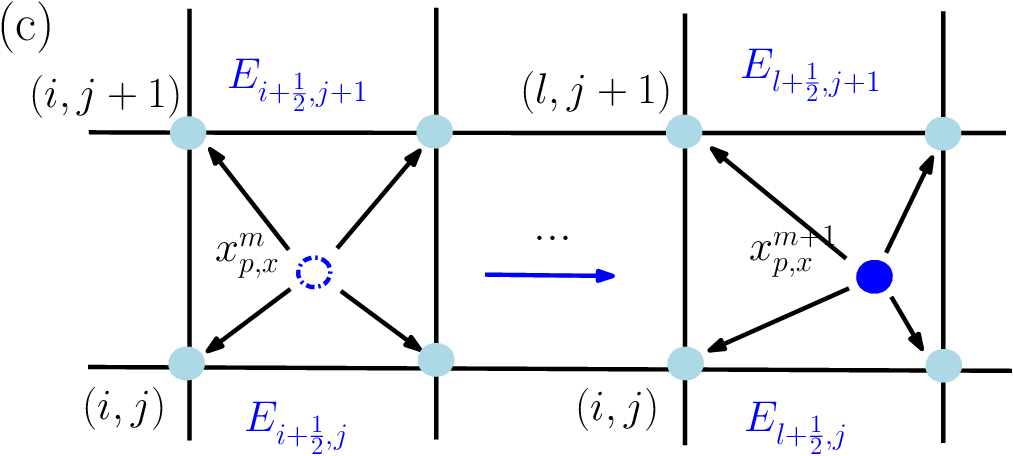}
            \vspace{2em}
			\caption{Electric-field update with particle motion. (a) The particle motion is split into two orthogonal moves. (b) Update $\bm{E}$ based on  Gauss's law within one cell. (c) Update $\bm{E}$ based on  Gauss's law when crossing the cell boundary.}
			\label{f:move1}
\end{figure}
After updating the position of each particle along each direction, the electric field along the path of motion in that direction is also immediately updated to satisfy the Gauss's law.
% This local update is very similar to solve the Poisson-Boltzmann equations~\cite{BSD:PRE:2009,ZWL:PRE:2011,Qiao2024SISC} where the ionic concentration is regarded as fractional particle moving along the mesh.
% The specific idea is detailed as follows.
For simplicity, we consider the support of the shape function $S$ as $[0,h]^2$. Consider one particle that moves from point $\bm{x}_p^{m}=({x}_{p,x}^{m},{x}_{p,y}^{m})$ to point $({x}_{p,x}^{m+1},{x}_{p,y}^{m})$ in $x$ direction as in~\reff{position1}. There exists a pair of index $(i,j)$, such that $ih\leq {x}_{p,x}^{m}<(i+1)h$ and $jh\leq {x}_{p,y}^{m}<(j+1)h$. Let us first assume that the movement distance is small enough so that the updated position is still within the same cell (cf. Figure~\ref{f:move1} (b)), meaning that $ih\leq {x}_{p,x}^{m+1}<(i+1)h$. Define the distances
$$d_{x,i}^m=| ih-{x}_{p,x}^{m}| ~\text{ and }~ d_{y,j}^m=|jh-{x}_{p,y}^{m}|.$$
Such movement causes the corresponding electric field $E_{i+1/2,j}$ and $E_{i,j+1/2}$ changed by
  $$
  E^*_{i+\frac{1}{2},j} \leftarrow E_{i+\frac{1}{2},j}- \frac{1}{\lambda^2}h w_p \left[ S_x(d_{x, i}^{m+1})-S_x(d_{x, i}^{m}) \right]S_y(d_{y, j}^{m}),
		$$
$$
		E^*_{i+\frac{1}{2},j+1} \leftarrow E_{i+\frac{1}{2},j+1}- \frac{1}{\lambda^2} h w_p \left[ S_x(d_{x, i}^{m+1})-S_x(d_{x, i}^{m}) \right]S_y\left(d_{y, j+1}^{m}\right).$$

If the movement distance is larger and the particle crosses the cell boundary, there exists an integer $l$ such that $lh\leq{x}_{p,x}^{m+1}<(l+1)h$; cf.~Figure~\ref{f:move1} (c). Here we assume $l>i$. The other case $l<i$ can be treated similarly.
All relevant electric fields along the path, i.e. $E_{i+1/2,j},\cdots$ $E_{l+1/2,j}$ and $E_{i+1/2,j+1},\cdots$ $E_{l+1/2,j+1}$ are updated accordingly based on the Gauss's law~\reff{DiscreteGauss}:
	\begin{equation}\label{m:v1}
		E^*_{i+\frac{1}{2},j} \leftarrow E_{i+\frac{1}{2},j}- \frac{1}{\lambda^2} h w_p \left[ -S_x(d_{x, i}^{m}) \right]S_y(d_{y, j}^{m}),
		\end{equation}
 \begin{equation}\label{m:vy1}
		E^*_{i+\frac{1}{2},j+1} \leftarrow E_{i+\frac{1}{2},j+1}- \frac{1}{\lambda^2} h w_p \left[ -S_x(d_{x, i}^{m}) \right]S_y\left( d_{y, j+1}^{m}\right),\end{equation}
  \begin{equation}\label{m:v2}
		E^*_{r+\frac{1}{2},j} \leftarrow E_{r+\frac{1}{2},j}- \frac{1}{\lambda^2}h w_p \left[ -S_x(d_{x, i}^{m}) +S_x\left(d_{x, i+1}^{m+1}\right)-S_x\left(d_{x, i+1}^{m}\right)\right]S_y(d_{y, j}^{m}),
		\end{equation}
  \begin{equation}\label{m:vy2}
		E^*_{r+\frac{1}{2},j+1} \leftarrow E_{r+\frac{1}{2},j+1}- \frac{1}{\lambda^2}h w_p \left[ -S_x(d_{x, i}^{m}) +S_x\left(d_{x, i+1}^{m+1}\right)-S_x\left(d_{x, i+1}^{m}\right) \right]S_y\left(d_{y, j+1}^{m}\right), \end{equation}
    for $r=i+1,\cdots,l-1,$ and
  \begin{equation}\label{m:v3}
		E^*_{l+\frac{1}{2},j} \leftarrow E_{i+\frac{1}{2},j}+ \frac{1}{\lambda^2}h w_p \left[ S_x\left(d_{x, l+1}^{m+1}\right)\right]S_y(d_{y, j}^{m}),
  \end{equation}
 \begin{equation}\label{m:vy3}
		E^*_{l+\frac{1}{2},j+1} \leftarrow E_{l+\frac{1}{2},j+1}+ \frac{1}{\lambda^2}h w_p \left[ S_x\left(d_{x, l+1}^{m+1}\right)\right]S_y\left(d_{y, j+1}^{m}\right) .\end{equation}
  The motions in other directions~\reff{position2}, moving from point $(x_{p,x}^{m+1},x_{p,y}^{m})$ to $\bm{x}_p^{m+1}=(x_{p,x}^{m+1},x_{p,y}^{m+1})$, are treated similarly. By doing so, we maintain the Gauss's law strictly by updating the electric field simultaneously with the particle motion, instead of solving the Poisson's or Amp\`ere's equation.

 %For simplicity, let  $\epsilon_{s} (\bm{x}) =1$ in Gauss's law~\reff{gauss}, although our new local method can accommodate a variable dielectric constant without any additional cost.
%Our new algorithm updates the relevant electric field according to Gauss's law while calculating the particle positions.

\subsection{Electric field correction}
   After above updates, there is one more condition that needs to be satisfied: the curl-free condition of the electric field. This condition can be achieved through a local curl-free relaxation algorithm that is of linear complexity~\cite{MR:PRL:2002,Qiao2023SIAP}. The starting point of the local relaxation algorithm is to minimize a convex electrostatic free energy $$\mathcal{F}_{\rm pot}(\bm{E})=\lambda^2 \int_\Omega {\bm{E}^2}/{2} \dd\bm{x},$$ subject to the constraint of the Gauss's law~\reff{Gauss}. It can be shown by the Lagrange multiplier method that the unique minimizer of the constrained optimization problem is the desired solution satisfying the curl-free condition ~\cite{Qiao2023SIAP}.
   After the updates in the previous section, we obtain Gauss's law satisfying fields $\bm{E}^*_h$, which shall be labeled as $\bm{E}^{(0)}_h$. In this section, we shall develop a Gauss's law preserving iterative method to minimize the free energy so that it approaches curl free on the Gauss's law preserving manifold.

   \begin{figure}
       \centering
		\includegraphics[height=4.4cm]{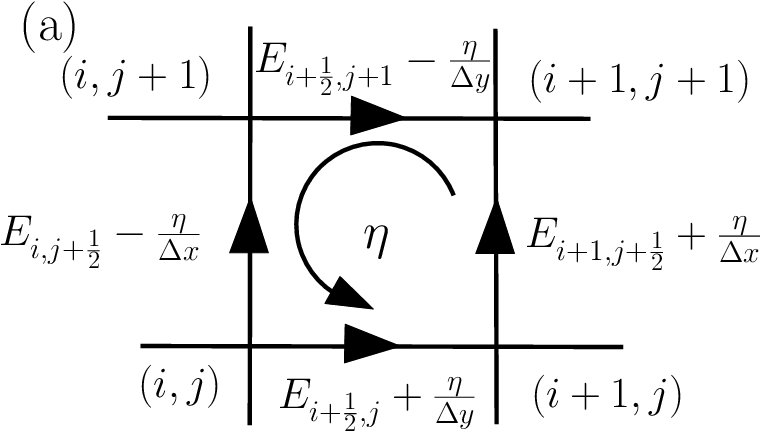}
		\hspace{-3em}
		\includegraphics[height=4.3cm]{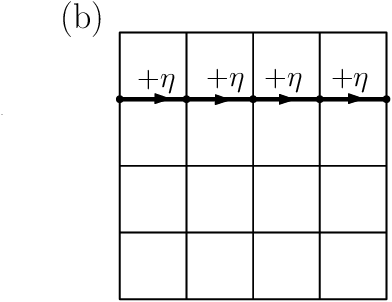}
        \vspace{2em}
		\caption{The correction of the electric field. (a): Local curl-free relaxation. (b): Line shift.}
  \label{f:maggs}
	\end{figure}

  The first step of correction involves looping over all the cells in the computational mesh. Let us take one cell as an example, as shown in Figure~\ref{f:maggs} (a). We increase the electric field components
$E_{i+1/2,j}$ and $E_{i+1,j+1/2}$ by a small increment $\eta$, while decrease $E_{i+1/2,j+1}$ and $E_{i,j+1/2}$ by $\eta$:
    \begin{equation}\label{update1}
        \begin{aligned}
    	&E_{i+\frac{1}{2},j}^{(q+1)}\leftarrow E_{i+\frac{1}{2},j}^{(q)}+\eta,~~~~~~E_{i+1,j+\frac{1}{2}}^{(q+1)}\leftarrow E_{i+1,j+\frac{1}{2}}^{(q)}+\eta, \\
    	&E_{i+\frac{1}{2},j+1}^{(q+1)}\leftarrow E_{i+\frac{1}{2},j+1}^{(q)}-\eta,~~~E_{i,j+\frac{1}{2}}^{(q+1)}\leftarrow E_{i,j+\frac{1}{2}}^{(q)}-\eta,
    \end{aligned}
    \end{equation}
    where the supscript $q$ denotes the iteration step in this correction phase.
    It is straightforward to check that the total fluxes entering and leaving each one of the four involved vertices remain unchanged, maintaining the Gauss's law on these four vertices.
    The associated change in the free energy due to the update of electric fields reads
		%\begin{footnotesize}
		\[
			\delta \mathcal{F}_{\rm pot}^{h}[\eta] = 2\lambda^2{\eta ^2} + \lambda^2{\eta}\left[ \left({{E}^{(q)}_{i+\frac{1}{2},j}}-{{E}^{(q)}_{i+\frac{1}{2},j+1}}\right)+ \left({{E}^{(q)}_{i+1,j+\frac{1}{2}}}-{{E}^{(q)}_{i,j+\frac{1}{2}}}\right) \right],	
		\]
    where $\mathcal{F}_{\rm pot}^{h}$ is the central-differencing approximation of $\mathcal{F}_{\rm pot}$ on the Yee mesh. It is easy to find that the associated free-energy change is minimized with
		\begin{equation}\label{eta1}
			\eta= - \frac{{{E}^{(q)}_{i+\frac{1}{2},j}}-{{E}^{(q)}_{i+\frac{1}{2},j+1}}+ {{E}^{(q)}_{i+1,j+\frac{1}{2}}}-{{E}^{(q)}_{i,j+\frac{1}{2}}}}{2}.
		\end{equation}

   The second step of correction involves line shift for the electric field to further minimize the free energy~\cite{Li2024convergence,jiang2018improved}. We follow the same idea that the flux entering and leaving each vertex should be kept the same. As depicted in Figure~\ref{f:maggs} (b), we increase electric fields on the $b$-th ($b=0,\cdots,N-1$) row by the same value $\eta$:
   \begin{equation}\label{update2}
       E_{i+\frac{1}{2},b}^{(q+1)}\leftarrow E_{i+\frac{1}{2},b}^{(q)}+\eta,~~  i=0,\cdots,N-1.
   \end{equation}
   The associated free-energy change reads
   % that will cause the energy modified by
   \[
		\delta \mathcal{F}_{\rm pot}^{h}[\eta] =\lambda^2{\eta ^2} \frac{N}{2} + \lambda^2{\eta} \sum_{i=1}^{N}E_{i+\frac{1}{2},b}^{(q)},
		\]
   which has a unique minimizer
		\begin{equation}\label{eta2}
		    \eta= - \frac{\sum_{i=0}^{N-1}E^{(q)}_{i+\frac{1}{2},b}}{N}.
		\end{equation}
  Similarly, we update the electric fields on the $a$-th ($a=0,\cdots,N-1$) column by
  \begin{equation}\label{update3}
       E_{a,j+\frac{1}{2}}^{(q+1)}\leftarrow E_{a,j+\frac{1}{2}}^{(q)}+\eta,~~ j=0,\cdots,N-1.
 \end{equation}
Then the minimizer of the associated free-energy change is given by
% $\eta$ has an explicit expression
\begin{equation}\label{eta3}
		    \eta= - \frac{\sum_{j=0}^{N-1}E^{(q)}_{a,j+\frac{1}{2}}}{N}.
\end{equation}

 It is stressed that the two local correction steps, \reff{update1},  \reff{update2}, and~\reff{update3}, are able to maintain the Gauss's law perfectly. Furthermore, the associated free-energy change is quadratic with an explicit expression for the minimizer. Therefore, the computational cost for each iteration step is linear.
%  We can find these two correction steps have specific and carefully designed forms~\reff{update1}, ~\reff{update2} and~\reff{update3}. It is actually these specific forms that
% simultaneously respect local properties and strictly ensure Gauss's law. Moreover, due to the convex nature of the energy, the minimum point is unique and has an explicit expression~\reff{eta1}, ~\reff{eta2} or~\reff{eta3} in each step. As we will see later, we can reduce the number of iterations to speed up the algorithm, because the impact of being curl-free for electric field is much less significant than that of Gauss's law.

Combining above steps described in sections 3.1, 3.2 and 3.3, we summarize the whole
local Gauss's law preserving PIC (GP-PIC) method in Algorithm \ref{alg:GP-PIC}.
  % the comprehensive algorithm is summarized as follows.

\begin{algorithm}[H] % [H] 强制算法在此处显示（不浮动）
\caption{Local Gauss's Law Preserving PIC (GP-PIC) Method}   % 算法名称
\label{alg:GP-PIC}  % 标签，用于交叉引用
\begin{algorithmic} % [1] 显示行号
\STATE \textbf{Initial:} $\bm{x}_p^m$, $\bm{v}_p^{m+1/2}$, $\bm{E}^m_h$, and tolerance $\tau$
\STATE \textbf{Update:} $\bm{x}_p^{m+1}$, $\bm{v}_p^{m+3/2}$, $\bm{E}^{m+1}_h$

\STATE \underline{\textbf{Step 1.}} Compute $\bm{x}_p^{m+1}$ and get $\bm{E}_h^*$
\FOR{$p = 1$ {to} $N_p$}
    \STATE Update the $x$-component position with~\reff{position1} and corresponding electric fields;
    \STATE Update the $y$-component position with~\reff{position2} and corresponding electric fields;
\ENDFOR

\STATE \underline{\textbf{Step 2.}} Set $q=0$ and $\bm{E}^{(0)}_h:=\bm{E}^{*}_h$. Correct $\bm{E}^{(0)}_h$ to get $\bm{E}^{m+1}_h$

\WHILE{the energy change $|\Delta \mathcal{F}_{\text{pot}}| \geq \tau$}
    \FOR{$i,j = 0$ \textbf{to} $N-1$}
    \STATE Update $\bm{E}^{(q)}_h$ by \reff{update1}: $\bm{E}^{(q_1)}_h \leftarrow \bm{E}^{(q)}_h$;
     \ENDFOR

\FOR{$b = 0$ {to} $N-1$}
    \STATE Update $\bm{E}^{(q_1)}_h$ by \reff{update2}: $\bm{E}^{(q_2)}_h \leftarrow \bm{E}^{(q_1)}_h$;
\ENDFOR
\FOR{$a = 0$ {to} $N-1$}
    \STATE Update $\bm{E}^{(q_2)}_h$ by \reff{update3}: $\bm{E}^{(q+1)}_h \leftarrow \bm{E}^{(q_2)}_h$;
\ENDFOR
\ENDWHILE

\STATE Set $\bm{E}_h^{m+1}=\bm{E}_h^{(q+1)}$

\STATE \underline{\textbf{Step 3.}}
Compute $\bm{v}_p^{m+1}$
\FOR{$p = 1$ {to} $N_p$}
    \STATE$$
v_{p,x}^{m+\frac{3}{2}}=v_{p,x}^{m+\frac{1}{2}}-{\Delta t  }  E_{p,x}^{m+1},~~~~~~v_{p,y}^{m+\frac{3}{2}}=v_{p,y}^{m+\frac{1}{2}}-{\Delta t } E_{p,y}^{m+1}.
$$
\ENDFOR

\end{algorithmic}
\end{algorithm}

\begin{theorem}
The proposed GP-PIC algorithm respects the discrete Gauss's law in the sense of
\[
\frac{E_{i+1/2,j}^{m+1}-E_{i-1/2,j}^{m+1}}{h}+\frac{E_{i,j+1/2}^{m+1}-E_{i,j-1/2}^{m+1}}{h}=\frac{n_0-n_{i,j}^{m+1}}{\lambda^2},
\]
provided that the Gauss's law at the previous time step holds:
\[
\frac{E_{i+1/2,j}^{m}-E_{i-1/2,j}^{m}}{h}+\frac{E_{i,j+1/2}^{m}-E_{i,j-1/2}^{m}}{h}=\frac{n_0-n_{i,j}^{m}}{\lambda^2}.
\]
\end{theorem}
\begin{proof}
 Let us check the Gauss's law at each vertex after each update. For Step 1 described in Section 3.2, we only check the case where the particle moves across the cell boundary. When the particle motion is within one cell, the proof is simpler and is omitted for brevity. With the motion as depicted in Figure~\ref{f:move1} (c) and the update~\reff{m:v1}-\reff{m:vy3} in electric fields, the discrete Gauss's law at the node $(i,j)$ becomes
    $$
    \begin{aligned}
        &\lambda^2\frac{E^*_{i+1/2,j}-E^m_{i-1/2,j}}{h} +\lambda^2\frac{E^m_{i,j+1/2}-E^m_{i,j-1/2}}{h}\\
        &= \lambda^2\frac{E^m_{i+1/2,j}-E^m_{i-1/2,j}}{h} +\lambda^2\frac{E^m_{i,j+1/2}-E^m_{i,j-1/2}}{h}-\omega_p [-S_x(d_{x,i}^m)]S_y(d_{y,j}^{m}) \\
        &= n_0-n^{m}_{i,j}-\omega_p [-S_x(d_{x,i}^m)]S_y(d_{y,j}^{m}) \\
        &= n_0-\left[n^{m}_{i,j}+ \Delta n_{i,j}\right],
    \end{aligned}
    $$
    where $\Delta n_{i,j}$ is the change of density at the point $(i,j)$ when a particle  moves from $\bm{x}_p^{m}=({x}_{p,x}^{m},{x}_{p,y}^{m})$ to $({x}_{p,x}^{m+1},{x}_{p,y}^{m})$. If $l=i+1$, the discrete Gauss's law at the node $(i+1,j)$ becomes
    $$
    \begin{aligned}
        & \lambda^2\frac{E^*_{i+3/2,j}-E^*_{i+1/2,j}}{h} +\lambda^2\frac{E^m_{i+1,j+1/2}-E^m_{i+1,j-1/2}}{h}\\
        &= n_0-n^{m}_{i+1,j} +w_p \left[  S_x\left(d_{x, i+2}^{m+1}\right)-S_x\left(d_{x, i}^{m}\right)\right]S_y(d_{y, j}^{m})\\
        &= n_0-n^{m}_{i+1,j}-w_p \left[  S_x\left(d_{x, i+1}^{m+1}\right)-S_x\left(d_{x, i+1}^{m}\right)\right]S_y(d_{y, j}^{m}) \\
        &= n_0-\left[n^{m}_{i+1,j}+ \Delta n_{i+1,j}\right],
    \end{aligned}
    $$
    where the identity $S_x(d_{x, i}^{m})+S_x(d_{x, i+1}^{m}) = S_x(d_{x, i+1}^{m+1})+S_x(d_{x, i+2}^{m+1})$ is used by the unity-integral property of the shape function.
    If $l>i+1$, then $S_x(d_{x, i+1}^{m+1})=0$. The discrete Gauss's law at the node $(i+1,j)$ becomes
    $$
    \begin{aligned}
        & \lambda^2\frac{E^*_{i+3/2,j}-E^*_{i+1/2,j}}{h} +\lambda^2\frac{E^m_{i+1,j+1/2}-E^m_{i+1,j-1/2}}{h}\\
        &= n_0-n^{m}_{i+1,j} -w_p \left[  S_x\left(d_{x, i+1}^{m+1}\right)-S_x\left(d_{x, i+1}^{m}\right)\right]S_y(d_{y, j}^{m})\\
        &= n_0-n^{m}_{i+1,j}-w_p \left[0-S_x\left(d_{x, i+1}^{m}\right)\right]S_y(d_{y, j}^{m})\\
        &= n_0-\left[n^{m}_{i+1,j}+ \Delta n_{i+1,j}\right].
    \end{aligned}
    $$
    The discrete Gauss's laws at the nodes $(r,j)$, $r=i+2,\cdots,l-1$ remain its form:
    $$
    \begin{aligned}
        & \lambda^2\frac{E^*_{r+1/2,j}-E^*_{r-1/2,j}}{h} +\lambda^2\frac{E^m_{r,j+1/2}-E^m_{r,j-1/2}}{h} \\
        &=\lambda^2\frac{E^m_{r+1/2,j}-E^m_{r-1/2,j}}{h} +\lambda^2\frac{E^m_{r,j+1/2}-E^m_{r,j-1/2}}{h}\\
        &=n_0-n_{r,j}^m.
    \end{aligned}
    $$
    The Gauss's law at $(l,j)$ becomes
    $$
    \begin{aligned}
        & \lambda^2\frac{E^*_{l+1/2,j}-E^*_{l-1/2,j}}{h} +\lambda^2\frac{E^m_{l,j+1/2}-E^m_{l,j-1/2}}{h}\\
        &= n_0-n^{m}_{l,j}+w_p \left[  S_x\left(d_{x, l+1}^{m+1}\right)-S_x\left(d_{x, i}^{m}\right)-S_x\left(d_{x, i+1}^{m}\right)\right]S_y(d_{y, j}^{m})\\
        &= n_0-n^{m}_{l,j}- w_p S_x\left(d_{x, l}^{m+1} \right)S_y(d_{y, j}^{m}) \\
        &=n_0-\left[n^{m}_{l,j}+ \Delta n_{l,j}\right],
    \end{aligned}
    $$
    where $S_x(d_{x, i}^{m})+S_x(d_{x, i+1}^{m})=S_x(d_{x, l}^{m+1})+S_x(d_{x, l+1}^{m+1})$ is used. The Gauss's law at the point $(l+1,j)$ becomes
    $$
         \begin{aligned}
        & \lambda^2\frac{E^m_{l+3/2,j}-E^*_{l+1/2,j}}{h} +\lambda^2\frac{E^m_{l+1,j+1/2}-E^m_{l+1,j-1/2}}{h}\\
        &= n_0-n^{m}_{l+1,j}-\omega_p S_x(d_{x,l+1}^{m+1}) S_y(d_{y,j}^{m}) \\
        &= n_0- \left[ n^{m}_{l+1,j} + \Delta n_{l+1,j} \right].
    \end{aligned}
     $$
     Thus, the electric field $\bm{E}^{*}$ satisfies Gauss's law in Step 1.     For the relaxation stage from $\bm{E}^{*}$ to $\bm{E}^{m+1}$ in Step 2, the Gauss's law preserving property is proved in~\cite{Li2024convergence}.
     In summary, the Gauss's law at each grid point is perfectly maintained in the whole GP-PIC method.
\end{proof}
\begin{rem}
We remark that the relationship between the electric field and particle motion positions is not constrained by the leapfrog scheme, and therefore, the GP-PIC method constitutes a general framework applicable to other temporal discretization schemes.
Also, the GP-PIC method can be coupled with other conservative schemes to achieve corresponding conservation properties. For example, energy conservation can be enforced by velocity correction~\cite{Liang2024JCP:energy}, while asymptotic preserving strategies~\cite{Li2023JCP} can be adopted to deal with quasi-neural limit. Moreover, by dimension splitting, the GP-PIC method can be directly extended to three dimensions.
% We must emphasize here that our proposed local Gauss's law satisfying method also has promising applicability to other temporal discretization, since we find an intrinsic expression for electric field and particle position that respects Gauss's law, making it powerful beyond the confines of the leapfrog scheme.
\end{rem}

\section{Numerical Results}\label{s:num}
 To demonstrate the effectiveness of the proposed GP-PIC method,
 % verify the conservation of Gauss's law of our proposed method,
 we conduct a series of 2D2V numerical experiments, including scenarios such as Landau damping, two-stream instability, and Diocotron instability. In this section, simulations are all performed on a computational domain
 % are defined as
 $\Omega_{\bm{x}}=[0,L]\times[0,L]$ with periodic boundary conditions.  To ensure the electro-neutrality condition required by the periodic boundary conditions, immobile ions are uniformly distributed in the background, i.e., only electrons are moving within the system.

 %The particle motion is calculated by using the``leapfrog" scheme, and the electric field is updated by using the proposed local Gauss's law preserving method.
 We compare the GP-PIC results with those obtained by solving the Amp\`{e}re equation and Poisson's equation for the update of the electric field.
 The ``leapfrog" scheme is employed for computing particle motions among all methods.
 % , and the results are compared with our method that satisfies local Gauss's law, while
 The method, labeled as ``VA-PIC'' in the following numerical tests, computes the electric field by Amp\`ere's equation
\begin{equation}\label{ampere}
    \bm{E}^{m+1}_h=\bm{E}^{m}_h+\Delta t \left( -\bm{J}_h^{m+\frac{1}{2}}+\bm{Q} \right)/\lambda^2 ,
\end{equation}
where we choose
 $\bm{Q}=0$, $\bm{J}_h^{m+\frac{1}{2}}=\sum_{p=1}^{N_p} w_p S(\bm{x}_h-\bm{x}_p^{m+1/2}) \bm{v}_p^{m+1/2}$ with $\bm{x}_p^{m+1/2}=(\bm{x}_p^{m+1}+\bm{x}_p^{m})/2$.
%\begin{aligned}
%	&\bm{x}^{m+1}=\bm{x}^m+{\Delta t} \bm{v}^{m+\frac{1}{2}}, \\
 %      \bm{E}^{m+1}_h=\bm{E}^{m}_h+\Delta t \left( -\bm{J}_h^{m+\frac{1}{2}}+\bm{Q} \right)/\lambda^2 ,
 % %    &\bm{v}^{m+\frac{3}{2}}=\bm{v}^{m+\frac{1}{2}}-{\Delta t } \bm{E}^{m+1}(\bm{x}_p),
%\end{aligned}
% $$
Then, the electric field is further relaxed by using Step 2 in the algorithm to achieve the curl-free condition.
 Notice that the Gauss's law is not exactly satisfied due to the interpolation used in $\bm{J}_h^{m+\frac{1}{2}}$ when solving the Amp\`ere's equation~\reff{ampere}.
 However, this method, which computes the electric field through Amp\`ere's equation, maintains the Gauss's law in Maxwell-Amp\`ere Nernst-Planck equations~\cite{Qiao2023SIAP,Qiao2022NumANP}, as it inherently enforces the local charge conservation through direct resolution of the Nernst-Planck equation.

 Alternatively, we can obtain the electric field being curl-free directly by solving the Poisson's equation for potential
\begin{equation}\label{pssn}
- \lambda^2 \text{ grad}_h \cdot \text{div}_h \phi^{m+1}=n_0-n^{m+1},
\end{equation}
where $\text{ grad}_h$ and $\text{div}_h$ are obtained by using central differencing on the Yee mesh. Then, the electric field is calculated by $\bm{E}^{m+1}_h=-\text{ grad}_h \phi^{m+1}$, which clearly satisfies both the Gauss's law and the curl-free condition with the accuracy dependent on the solver of the linear system~\reff{pssn}.
The GP-PIC method demonstrates inherent consistency with solutions from the Poisson's equation due to its enforcement of Gauss's law perfectly, despite the the accuracy of being curl-free being controlled by tolerance $\tau$ in the GP-PIC algorithm.
% Numerical experiment results in this section show that although the proposed method does not strictly meet the curl-free condition, it is completely consistent with the Poisson's equation results in other metrics, such as electric energy, total energy and phase image of particle.

\subsection{Landau damping}

\begin{figure}[t!]
		\centering
	\includegraphics[scale=0.5]{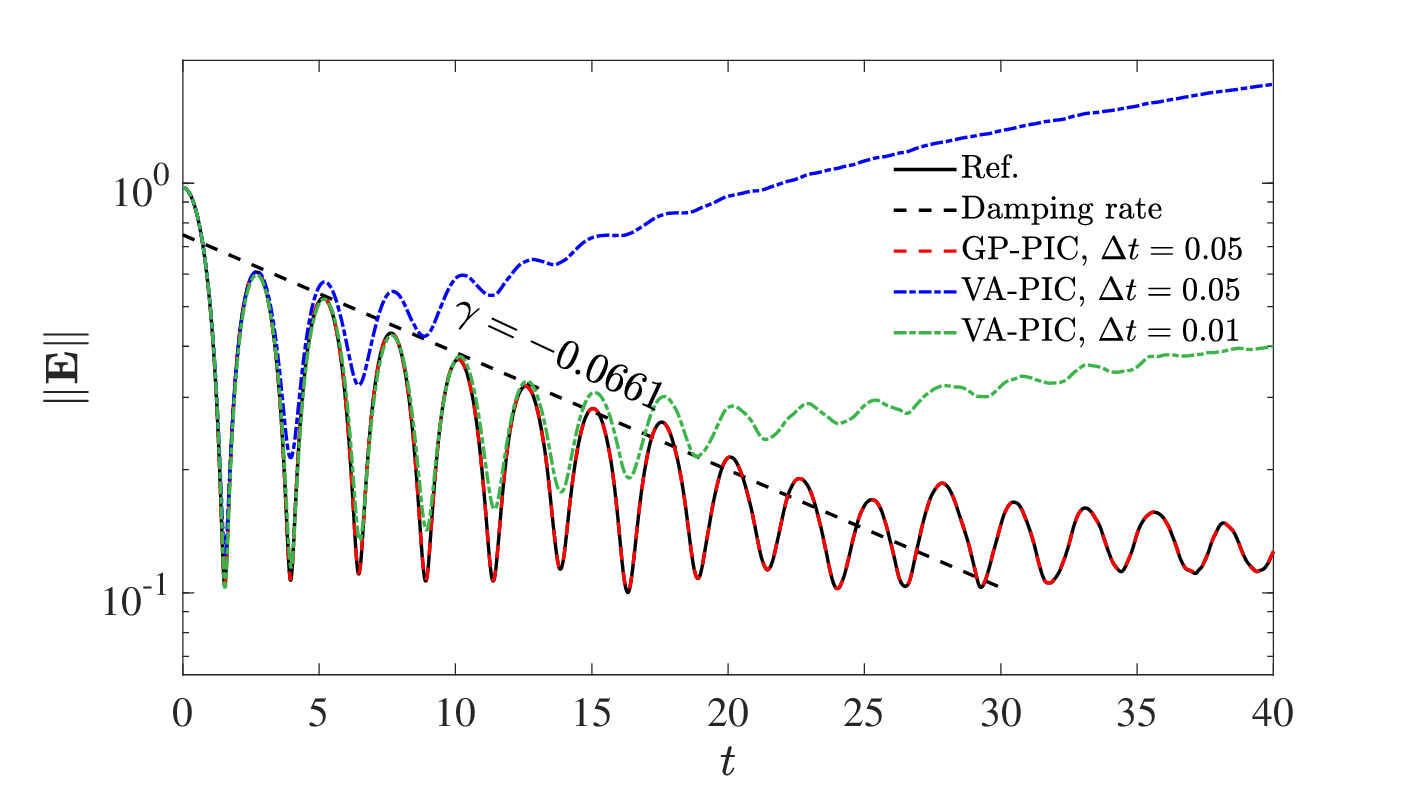}	\caption{The square root of electric energy in the simulation of Landau damping with different numerical methods. The reference solution, labeled as ``Ref.'' is obtained by solving the Poisson's equation for the electric field with $\Delta t=0.01$.}
  \label{f:landau1}
	\end{figure}

Landau damping is a widely tested PIC simulation~\cite{Filbet2001JCP,Besse2003JCP,Liu_Xu_2017,Medaglia2023JCP,Nicolas2009JCP}. We  initialize the position and velocity of particles by sampling the distribution
\[
	f(\bm{x},\bm{v},t=0)=\frac{1}{2\pi L^2} \left(1+\alpha \cos(\bm{k} \cdot\bm{x})\right) e^{-\frac{\bm{v}^2}{2}},
\]
 where $\alpha=0.05$, $\bm{k}=(0.4,0)^T$ and $L=2\pi/0.4$. We choose the mesh size $N_x=N_y=32$ and the number of particle $N_p=6.4\times 10^5$. The solution obtained by solving the Poisson's equation for the electric field with a time step $\Delta t=0.01$ serves as the reference solution. It is compared with the results obtained from our local method and solving the Amp\`ere's equation with a timestep $\Delta t=0.05$. In Figure~\ref{f:landau1}, we present the evolution of the square root of the electric energy given by
 \[
 \|\bm{E}\|= \left[ \sum_{i,j}^{N^2} h^2
 \left(|{E}_{i+\frac{1}{2},j}|^2+  |E_{i,j+\frac{1}{2}}|^2\right) \right] ^{\frac{1}{2}},
 \]
 %and total energy given by
 %\[
 %F=\frac{\lambda^2}{2} \|\bm{E}\| +\sum_{p=1}^{N_p} \omega_p \frac{\bm{v}_p^2}{2} ,
 %\]
using the three different methods.
 The square of electric energy shown in Figure~\ref{f:landau1} demonstrates that the results from our new methods agree with the reference solution very well and show a damping rate close to $\gamma=-0.0661$ that is predicted by theory~\cite{Liu_Xu_2017}. In contrast, the result computed by the Amp\`ere's equation, which may not satisfy the Gauss's law, exhibits relatively large errors, especially in late stage of long-term simulations. Its performance can be significantly improved with smaller time stepsizes, but remaining worse than the Gauss's law satisfying schemes for long time. These results demonstrate the importance of maintaining the Gauss's law during the simulations.

\begin{figure}[t!]
		\centering
  \includegraphics[scale=0.48]{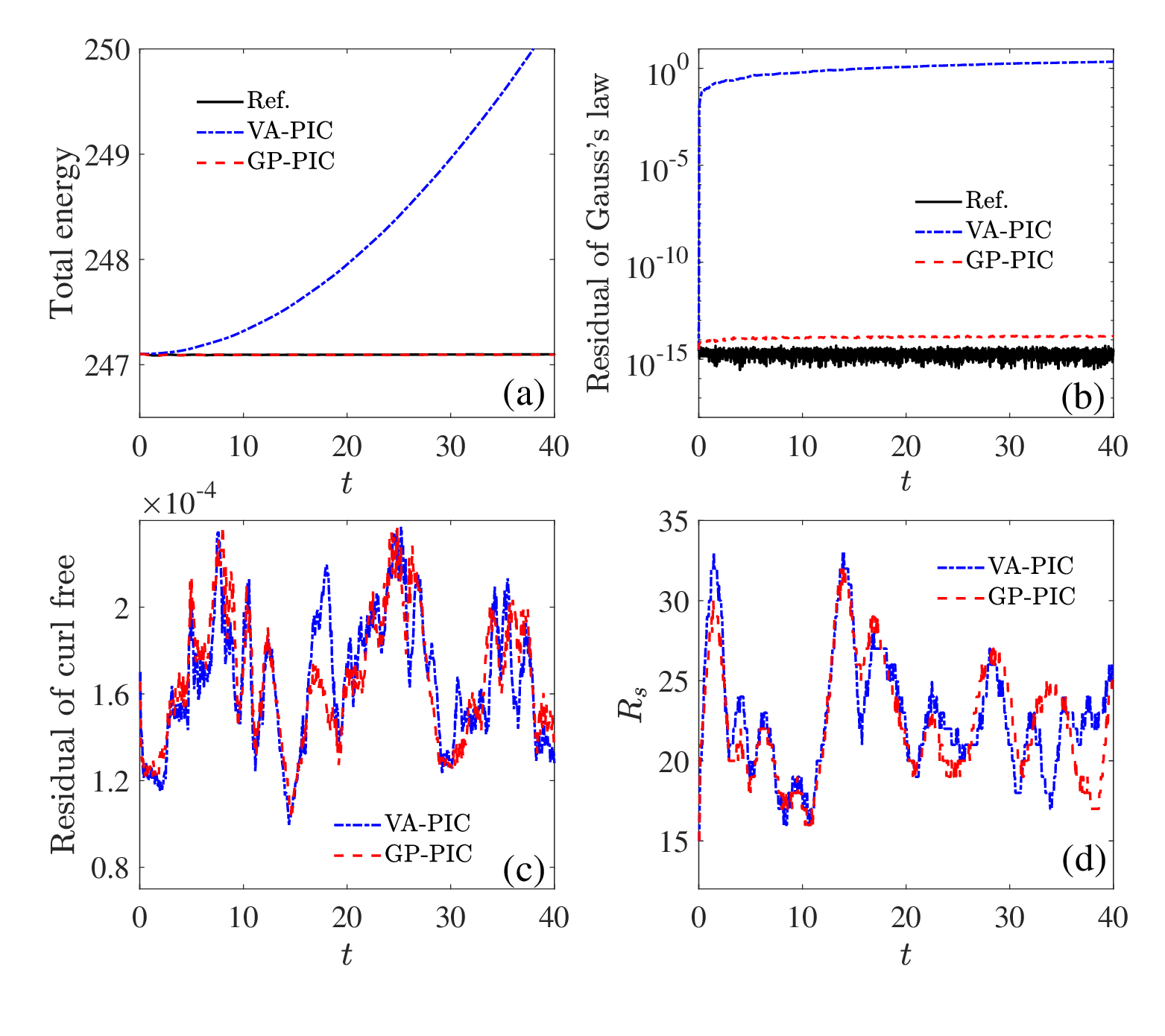}
  \caption{Simulation results for the case of the Landau damping with $\Delta t=0.05$ and $\tau=10^{-7}$: (a) the total energy; (b) the
residual error of the Gauss’s law; (c) the residual error of the curl-free condition; (d) The number
of iteration steps, $R_s$, used in relaxation stage. }
  \label{f:landau2}
	\end{figure}

 Figure~\ref{f:landau2} presents the performance of the methods in preserving the total energy,
 the Gauss's law, the curl-free condition in terms of residual errors and the number of iteration steps used in the correction stage.
 % The diagrams of error of Gauss's law and curl free are also exhibited in .
 The total energy is discretized by
 \[
 F=\frac{\lambda^2}{2} \|\bm{E}\|^2 +\sum_{p=1}^{N_p} \omega_p \frac{\bm{v}_p^2}{2}.
 \]
 Although we have not theoretically establish the energy conservation for the proposed local method, Figure~\ref{f:landau2} (a) demonstrates that the total energy is preserved as well as the results obtained by the Poisson's equation. In contrast, the total energy obtained by the Amp\`ere's equation keeps growing in the simulations. In terms of energy conservation, the energy-conserving correction strategies~\cite{Liang2024JCP:energy} can be combined to improve the performance of the numerical methods.
 From Figure~\ref{f:landau2} (b), one can observe that the residual errors of the Gauss's law for the electric field obtained from the proposed local method and Poisson's equation are both close to the machine precision, due to round-off error. However, the residual error of the Gauss's law computed by the Amp\`ere's equation increase significantly in simulations.
 As the electric field calculated using Poisson's equation with Yee mesh strictly satisfies the curl-free condition, we have omitted the Poisson's results in Figure~\ref{f:landau2} (c). %As shown in Figure~\ref{f:landau2} (c), the curl-free condition is maintained very well for
 % It is strictly satisfy curl-free requirement in
 The residual error of curl free calculated by the proposed local method and the Amp\`ere's equation increases in initial a few steps and is well controlled in the entire simulations via the stopping criterion $\tau=10^{-7}$ we set in both algorithms. By setting the same stopping criterion, the iteration steps in the correction stage are on the same level for GP-PIC method and VA-PIC method as displayed in Figure~\ref{f:landau2} (d).

 % The curl errors of the other two methods are more significant compared with Poisson's equation, but they are on the same order of magnitude, which is realized by setting the stopping criterion for the correction steps to ensure a fair comparison.

\subsection{Two-stream instability}

\begin{figure}[t!]
		\centering
	\includegraphics[scale=0.5]{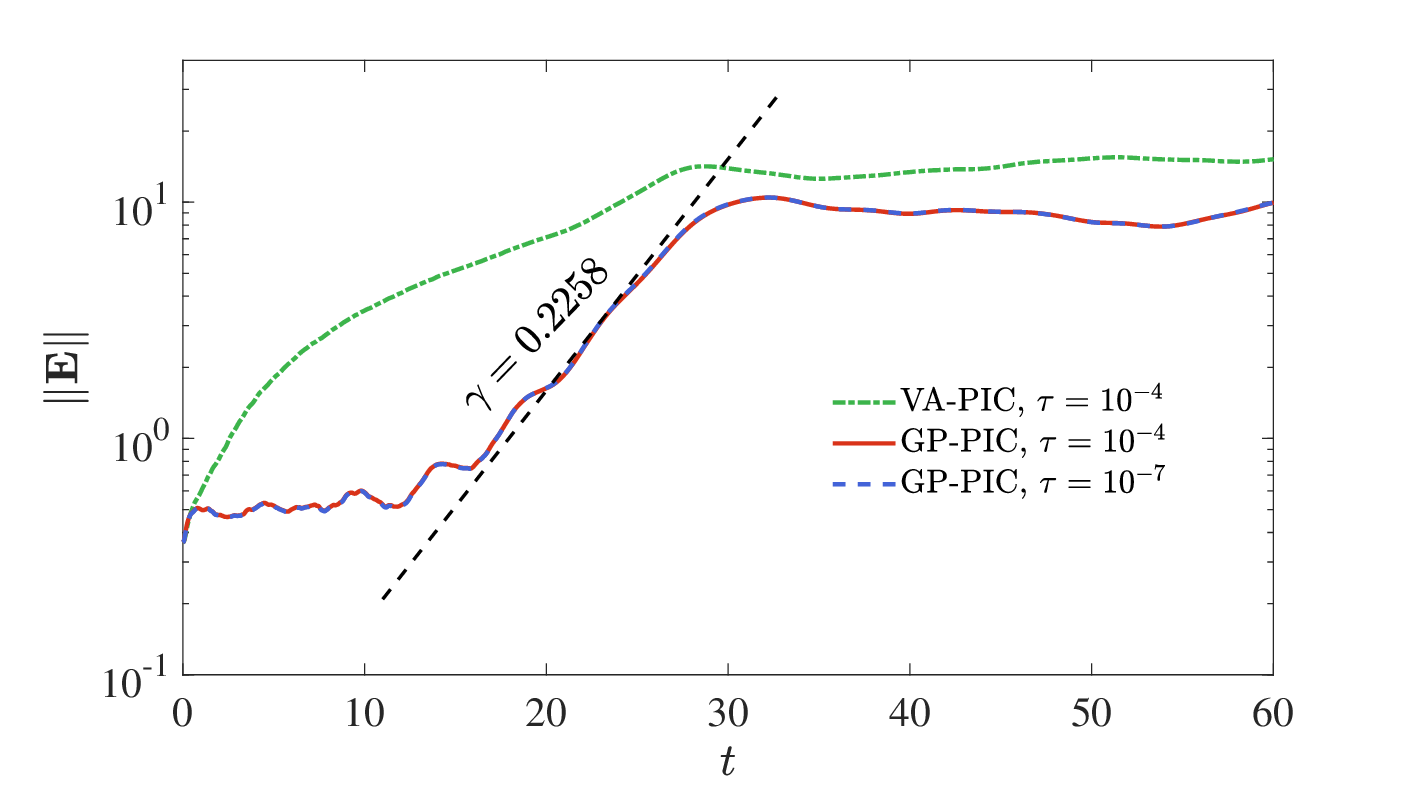}
		\caption{The square root of electric energy in the simulation of two-stream instability with $\Delta t=0.05$. The black dashed line represents a theoretical growth rate $\gamma=0.2258$. Two different stopping criteria $\tau=10^{-4}$ and $\tau=10^{-7}$ are considered in the curl-free relaxation stage. }
  \label{f:twostream1}
	\end{figure}

Two-stream instability is a typical phenomenon that has been well investigated by many studies~\cite{Besse2003JCP,Nicolas2009JCP,Medaglia2023JCP,Ricketson2024explicit}. In the two-stream instability test,  two beams of particles, each with an average velocity in opposite directions, move through space. Such a plasma system is inherently unstable and these two beams eventually merge to form a vortex structure. During the simulation, energy gradually transfers from the particles to the electric field, resulting in a decrease in kinetic energy and an increase in electric energy.
The initial particle position and velocity is obtained by sampling the distribution
    \[
	f(\bm{x},\bm{v},t=0)=\frac{1}{2\pi L^2} \left(1+\alpha \cos(\bm{k} \cdot\bm{x})\right) \left[ \frac{1}{2} e^{-\frac{(\bm{v}-\bm{v}_d)^2}{2}}+\frac{1}{2} e^{-\frac{(\bm{v}+\bm{v}_d)^2}{2}} \right],
	\]
 where $\alpha=0.003$, $\bm{k}=(0.2,0)^T$, $L=2\pi/0.2$ and $\bm{v}_d=(2.4,0)^T$. We choose the mesh size $N_x=N_y=64$ and number of particle $N_p=6.4\times 10^5$. The square root of electric energy is displayed in Figure~\ref{f:twostream1}, which reveals that the proposed local Gauss's law preserving method is consistent with previous works~\cite{Medaglia2023JCP}, with a theoretical
growth rate of electric energy being $\gamma=0.2258$. Note that our proposed GP-PIC method agrees well with Poisson's result, so we omit the figures of Poisson's result in this example.

\begin{figure}[t!]
		\centering
	\includegraphics[scale=0.48]{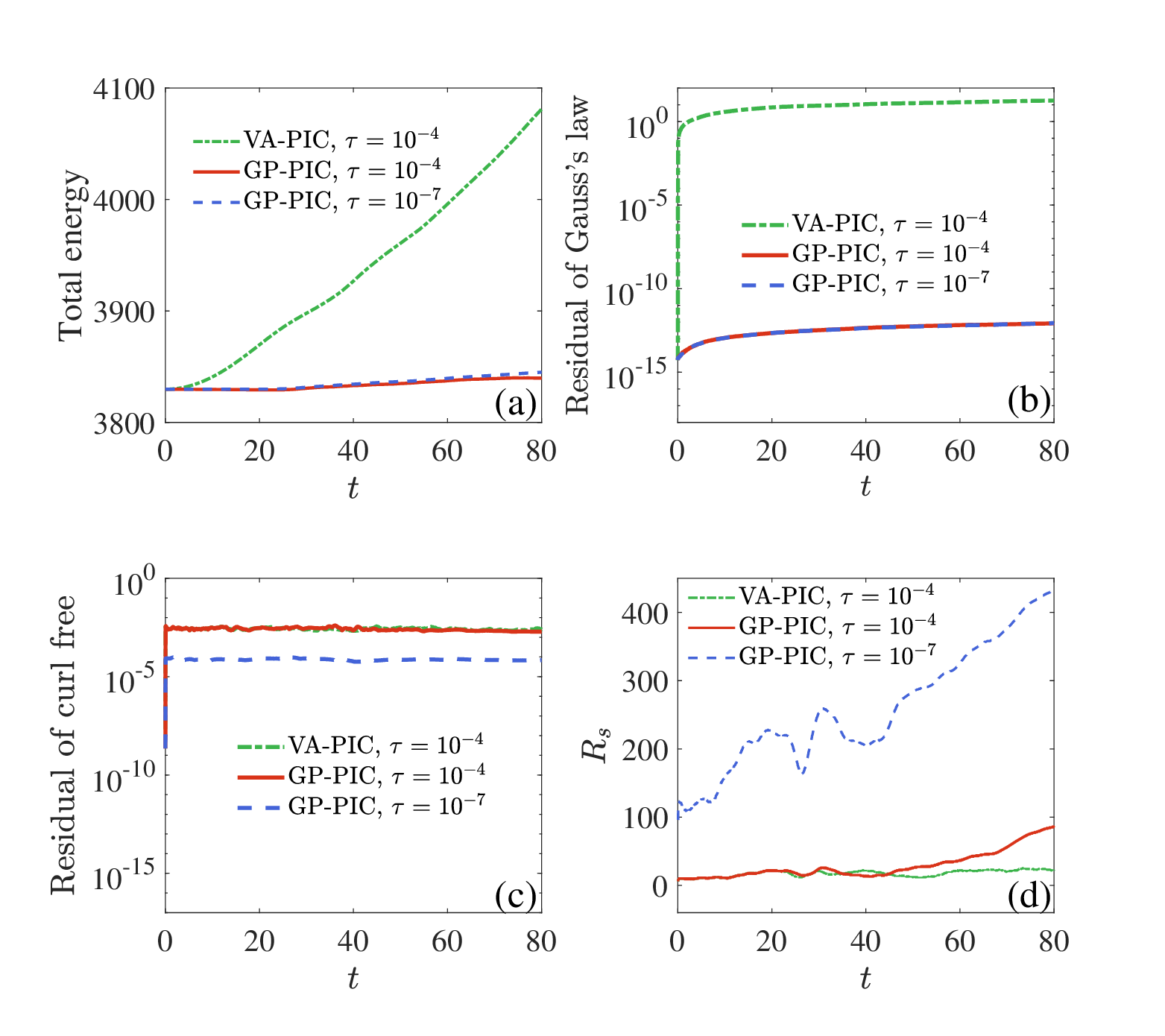}
		\caption{ Simulation results for the case of the two-stream instability with $\Delta t=0.05$: (a) the total energy; (b) the residual error of the Gauss's law; (c) the residual error of the curl-free condition; (d) The number of iteration steps, $R_s$, used in relaxation stage.}
  \label{f:twostream2}
	\end{figure}

Figure~\ref{f:twostream2} displays the total energy, the residual error of the Gauss's law, curl-free condition, and the number of relaxation steps in the correction stage.
We observe the similar results as Landau damping that the total energy computed by VA-PIC method deviates significantly from the initial energy as time evolves, whereas our proposed GP-PIC method exhibits significantly smaller deviations, displayed in Figure~\ref{f:twostream2} (a). As previously discussed, while energy-conserving scheme can indeed be designed, its rigorous implementation lies beyond the scope of this work.
% We can get the error of the Gauss's law and curl free, and the relaxation step in the correction steps from Figure~\ref{f:twostream2}.
 From Figure~\ref{f:twostream2} (b), one can find that the Gauss's law is perfectly satisfied in the proposed local methods. However, it is not true for the results obtained by the Amp\`ere's equation.
 To test the effect of the stopping criteria $\tau$, we perform simulations with $\tau=10^{-4}$ and $\tau=10^{-7}$.
 It is of interest to find that with a more stringent stopping criterion $\tau$, the residual error of the Gauss's law does not change much, as expected. But the number of iteration steps, $R_s$, in the correction stage of the proposed algorithm increases and the residual error of the curl-free condition effectively reduces one order of magnitude, as revealed by Figure~\ref{f:twostream2} (c) and (d).
  \begin{figure}[t!]
		\centering
	\includegraphics[scale=0.42,trim={0cm 2cm 0cm 0cm}, clip]{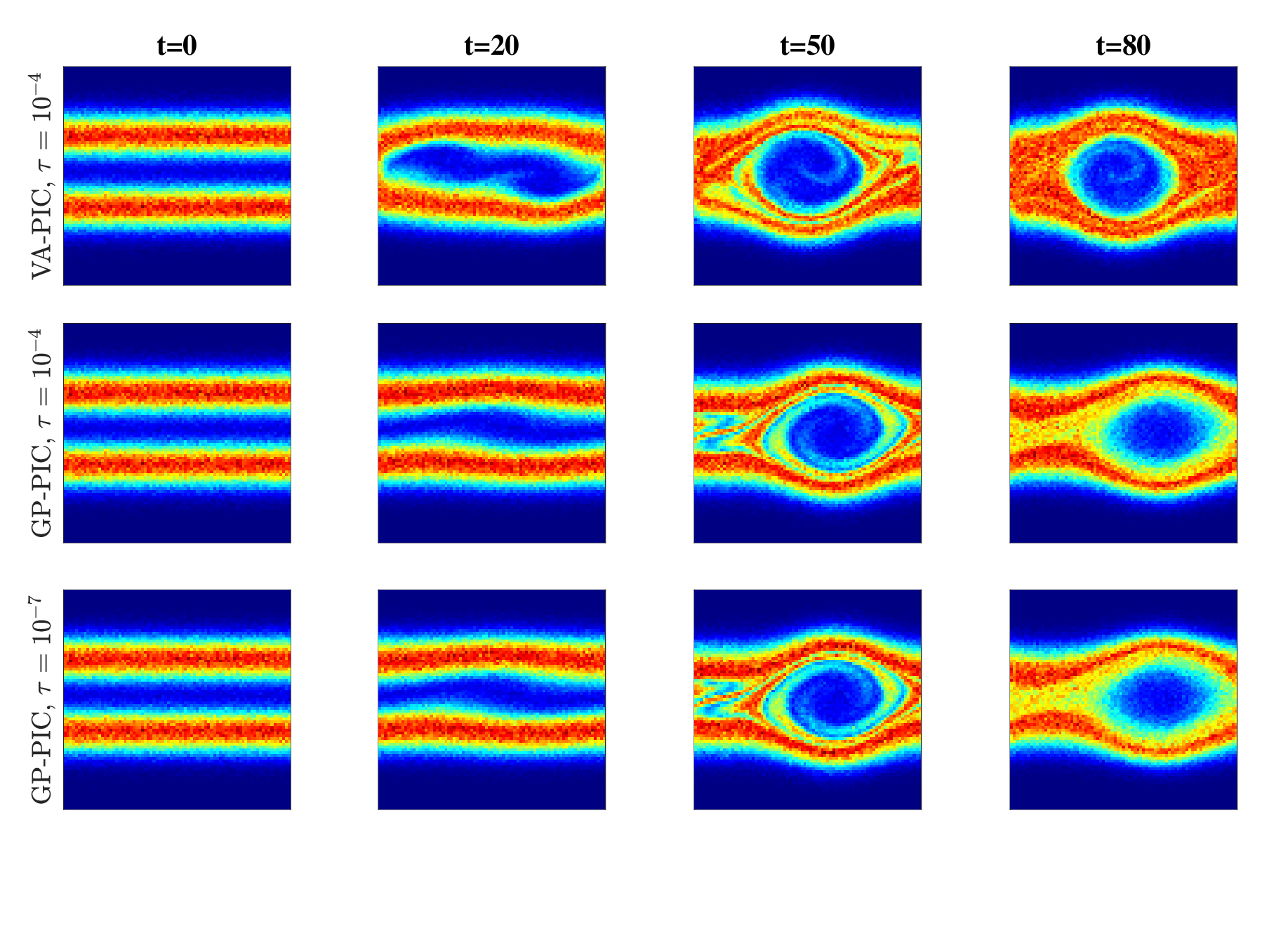}
		\caption{ The $x-v_x$ phase space distribution for the case of two-stream instability with $\lambda = 1$ and $\Delta t= 0.05$ at time $t=0,~20,~50,~80$. The first row shows the results computed by the Amp\`ere's equation. The second row ($\tau=10^{-4}$) and the third row ($\tau=10^{-7}$) are obtained by the proposed method with different stopping criteria in the correction stage.}
  \label{f:twostreamphase}
	\end{figure}

 As shown in Figure~\ref{f:twostreamphase}, we also study the $x$-$v_x$ phase space distribution of the system with different methods. A violation of the Gauss's law results in a prominent error accumulation for long-term simulations, as revealed by the results obtained from the Amp\`ere's equation in the first row of Figure~\ref{f:twostreamphase}. The phase space distribution exhibits vortex structures that are significantly different from those obtained using the local method.  Furthermore, the phase space distribution does not change much with a more stringent stopping criterion $\tau$, indicating that the curl-free condition is not as crucial as the Gauss's law in long-term simulations.

% Tightening iteration stopping criterion $\tau$, accordingly reducing the error of curl free and increasing the number of iteration steps in the correction step, presented in subplot (b) and (c), does not significantly make the result better, indicating that the curl-free condition is not as crucial as Gauss's law.

% A violation of Gauss's law results in a prominent error accumulation for long-term simulations, as shown by the results obtained from Amp\`ere's equation.
% This assertion is further evidenced by the phase diagram presented in Figure~\ref{f:twostreamphase}.
% The $x$-$v_x$ phase space distribution images calculated from the Ampère's equation, which does not satisfy Gauss's law, show vortex structures that are significantly different from those obtained using the local method.

    \subsection{Diocotron instability}
Finally, we perform a variation of Diocotron instability to study the electron vortices when exposed to an external magnetic field. This instability has been well studied numerically in many works~\cite{Filbet2016SINUM,Ricketson2016sparse,GU2022JCP,Ricketson2023CPC}. The initial electron distribution follows
$$
f(x,y,\bm{v},t=0)=\frac{C}{2\pi} e^{-\frac{\bm{v}^2}{2}} e^{-\frac{(r-L/4)^2}{2(\delta L)^2} },~~ r=\sqrt{(x-L/2)^2+(y-L/2)^2},
$$
where $\delta=0.03$, $L=22$ and $C$ is a normalization parameter. Here, a uniform magnetic field $\bm{B}=(0,0,B_z)$ is imposed, and we choose $B_z=5$ and $B_z=15$. The mesh size is $N_x=N_y=64$, timestep is $\Delta t=0.01$ and the number of particles is taken as $N_p=10^6$ in the simulation. In this case, with finite-size particle approximation~\reff{fapp} for the distribution function, the motions of electron follow
\begin{align}
		&\frac{\dd \bm{x}_p}{\dd t} =\bm{v}_p , \\
		& \frac{\dd \bm{v}_p}{\dd t} = - \left[\bm{E}(\bm{x}_p) + \bm{v}_p\wedge \bm{B}\right],
\end{align}
for $p=1,\cdots,N_p$.

    \begin{figure}[t!]
		\centering
        \includegraphics[width=13cm,trim={0cm 2cm 1cm 1cm}, clip]{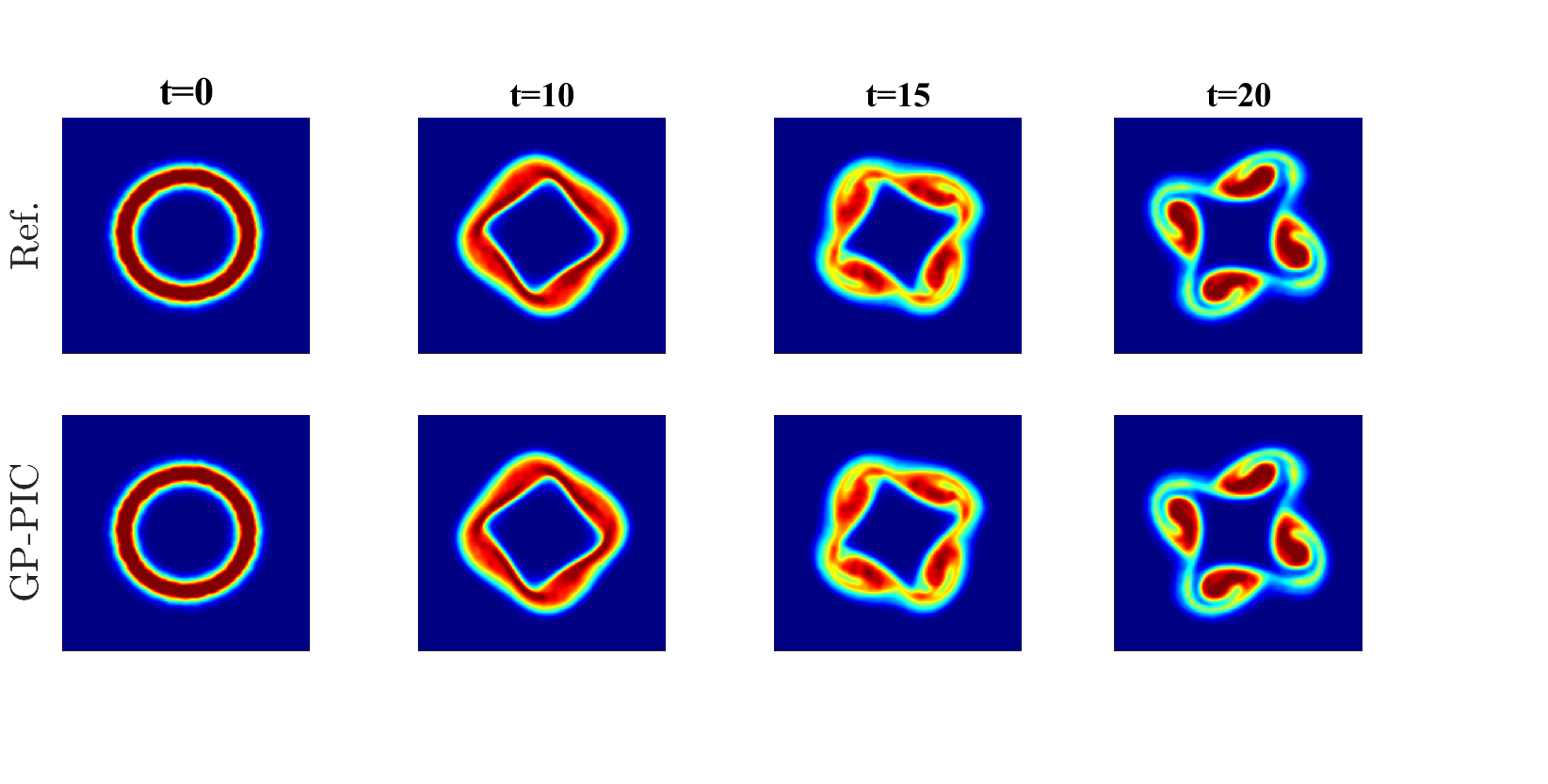}
		\caption{ The distribution of electron in the case of Diocotron instability at time $t=0,~10,~15,~20$ with $B_z=5$. The reference solution, labeled as ``Ref.'' is obtained by solving the Poisson’s
equation for the electric field.}
  \label{f:diocotronphase}
	\end{figure}
    \begin{figure}[t!]
		\centering
	 \includegraphics[width=13cm,trim={0cm 2cm 1cm 1cm}, clip]{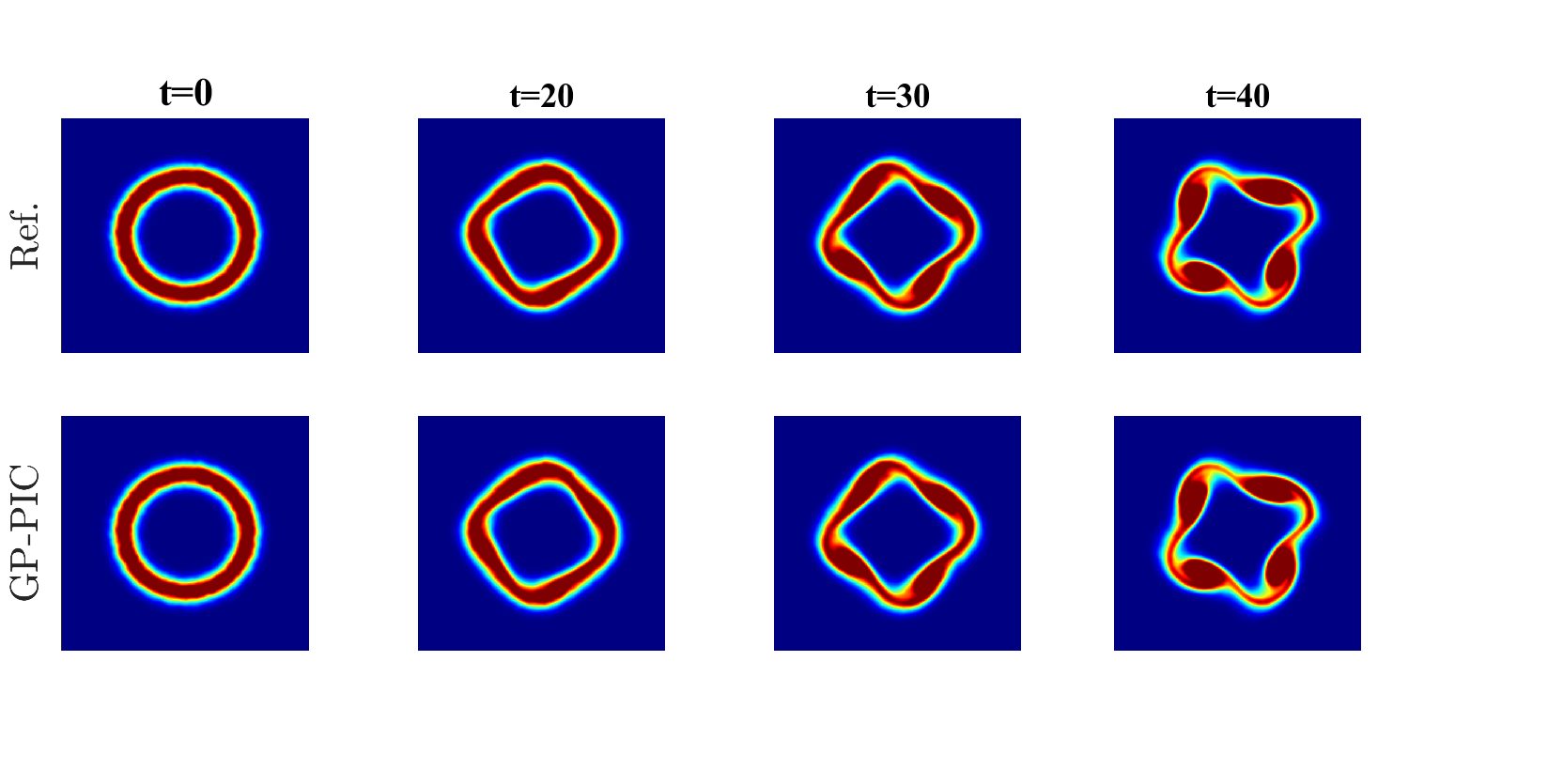}
		\caption{ The distribution of electron in the case of Diocotron instability at time $t=0,~20,~30,~40$ with $B_z=15$. The reference solution, labeled as ``Ref.'' is obtained by solving the Poisson’s
equation for the electric field.}
  \label{f:diocotronphase2}
	\end{figure}
    %As evidenced by Figure~\ref{f:diocotronphase} and Figure~\ref{f:diocotronphase2}, the magnetic field in the z-direction causes the electron distribution to form a vortex structure, which is effectively captured by the proposed method.  By setting $B_z=15$ in Figure~\ref{f:diocotronphase2}, we obtain four clear vortex structures similar to those reported in the work~\cite{Ricketson2016sparse}. Reducing $B_z$ to $5$ further enhances the clarity of these structures, as demonstrated in \cite{GU2022JCP}.\zhou{Check here!}

    As evidenced by Figure~\ref{f:diocotronphase} and Figure~\ref{f:diocotronphase2}, the proposed method can effectively capture a vortex structure of the electron distribution due to
    the magnetic field in the $z$-direction.
    By setting $B_z=5$ in Figure~\ref{f:diocotronphase}, four well-resolved vortex structures emerge prominently in the simulation at $t = 20$. When the strength of magnetic field is increased to $B_z = 15$ in Figure~\ref{f:diocotronphase2}, it takes longer time to form vortex structures, being consistent with the results reported in the work~\cite{Ricketson2016sparse}. Furthermore, the vortices exhibit diminished features compared to the $B_z = 5$ configuration, a phenomenon demonstrated in \cite{GU2022JCP}.

\section{Conclusions}\label{s:Con}
The importance of rigorously maintaining the Gauss's law in plasma simulations cannot be  overemphasized‌. This work has proposed a Gauss's law preserving method for the electric field updating in particle-in-cell simulations in  the electrostatic limit.  The proposed method updates the electric field locally according to the electric fluxes induced by particle motion through splitting the motion into sub-steps along each dimension of the computational mesh, resulting a method that can enforce the Gauss's law exactly. To obtain a curl-free electric field, a local update scheme has been developed to further correct the electric field by relaxing the electric-field free energy subject to the Gauss's law. Theoretical analysis has confirmed that the proposed method maintains the Gauss's law up to round-off precision. Numerical performance on classical benchmarks, including the Landau damping, two-stream instability and Diocotron instability, has evidenced the advantages of the proposed method. It is expected that the local nature of the proposed method makes it a promising tool in parallel simulations of large-scale plasma. In addition, the feature that electric fields are locally updated according to the electric fluxes induced by particle motion can be further employed to construct a fully implicit discretization that allows a large time stepping size, which has great potential in efficient simulations of the quasineutral cases.

% This work emphasizes the critical importance of preserving fundamental conservation laws in particle-in-cell simulations to ensure numerical accuracy and physical fidelity, with a particular focus on the conservation of Gauss's law. This conservation law not only serves as the cornerstone of self-consistent electrostatic interactions, but also the core of the novel framework we proposed. In our method, the particle motion is divided into two (in two dimensions) or three (in three dimensions) orthogonal motions parallel to the coordinate axes. Then we use a unified field-particle advancement scheme that simultaneously updates electric fields according to the trajectory of particle movement. The electric fields are further corrected with a Gauss's law satisfying relaxation. Theoretical analysis and numerical benchmarks confirm that this approach maintains Gauss's law to machine precision, with violations bounded at round-off error levels.  While this paper primarily shows the commonly used leapfrog scheme, our novel framework is independent from specific time discretization formats, along with its ability to integrate additional conservation mechanisms such as energy conservation and asymptotic preserving properties, demonstrating its remarkable flexibility and strong potential for future applications.

\section*{Acknowledgements}
This work is supported by the CAS AMSS-PolyU Joint Laboratory of Applied Mathematics
(Grant No. JLFS/P-501/24). Z. Qiao is partially supported by the Hong Kong Research Grants Council RFS grant RFS2021-5S03, NSFC/RGC Joint Research Scheme N\_PolyU5145/24, GRF grants 15302122 and 15305624. Q. Yin is partially supported by the Hong Kong Polytechnic University Postdoc Matching Fund Scheme 4-W425. The work of Z. Xu is partially supported by  NSFC (grants No. 12325113 and 12426304) and National Key R\&D Program of China (grant No. 2024YFA1012403). The work of S. Zhou is partially supported by National Key R\&D Program of China 2023YFF1204200.

\bibliographystyle{plain}
\bibliography{refbib}

\end{document}